\documentclass{article}
\usepackage{arxiv}

\usepackage[utf8]{inputenc} % allow utf-8 input
\usepackage[T1]{fontenc}    % use 8-bit T1 fonts
\usepackage{hyperref}       % hyperlinks
\usepackage{url}            % simple URL typesetting
\usepackage{booktabs}       % professional-quality tables
\usepackage{amsfonts}       % blackboard math symbols
\usepackage{nicefrac}       % compact symbols for 1/2, etc.
\usepackage{microtype}      % microtypography
\usepackage{lipsum}		% Can be removed after putting your text content
\usepackage{graphicx}
\usepackage{natbib}
\usepackage{doi}

\usepackage{amsmath,amsfonts,amssymb,amsxtra,xfrac, yhmath}
\usepackage[mathscr]{eucal}
\usepackage{wrapfig}
\usepackage{verbatim}
\usepackage{color}
\usepackage{enumitem}
\usepackage{amsthm}

\usepackage{scalerel}
\usepackage{tikz}
\usetikzlibrary{svg.path}

\definecolor{orcidlogocol}{HTML}{A6CE39}
\tikzset{
  orcidlogo/.pic={
    \fill[orcidlogocol] svg{M256,128c0,70.7-57.3,128-128,128C57.3,256,0,198.7,0,128C0,57.3,57.3,0,128,0C198.7,0,256,57.3,256,128z};
    \fill[white] svg{M86.3,186.2H70.9V79.1h15.4v48.4V186.2z}
                 svg{M108.9,79.1h41.6c39.6,0,57,28.3,57,53.6c0,27.5-21.5,53.6-56.8,53.6h-41.8V79.1z M124.3,172.4h24.5c34.9,0,42.9-26.5,42.9-39.7c0-21.5-13.7-39.7-43.7-39.7h-23.7V172.4z}
                 svg{M88.7,56.8c0,5.5-4.5,10.1-10.1,10.1c-5.6,0-10.1-4.6-10.1-10.1c0-5.6,4.5-10.1,10.1-10.1C84.2,46.7,88.7,51.3,88.7,56.8z};
  }
}

\newcommand\orcid[1]{\href{https://orcid.org/#1}{\mbox{\scalerel*{
\begin{tikzpicture}[yscale=-1,transform shape]
\pic{orcidlogo};
\end{tikzpicture}
}{|}}}}
\usepackage{hyperref}
\usepackage{msc}

\newtheorem{theorem}{Theorem}[section]%
\newtheorem{lemma}{Lemma}[section]%
\newtheorem{corollary}{Corollary}[section]%
\newtheorem{example}{Example}[section]%
\newtheorem{remark}{Remark}[section]%
\newtheorem{definition}{Definition}[section]%

\DeclareMathOperator{\Av}{\operatorname{Av}}

\renewcommand{\d}{{{\mathrm d}}}

\title{Quantum Random Walks and Quantum Oscillator in an Infinite-Dimensional Phase Space}

\author{ 
  \orcid{0009-0000-5636-1962} Vladimir Busovikov\\
  Department of Mathematical Methods\\ for Quantum Technologies\\
  Steklov Mathematical Institute \\of Russian Academy of Sciences\\
  Russia, Moscow, Gubkina Str. 8 \\
  \texttt{treonon38@mail.ru} 
	\And
  \orcid{0000-0001-8290-8300} Alexander Pechen \\
  Department of Mathematical Methods \\for Quantum Technologies\\
  Steklov Mathematical Institute \\of Russian Academy of Sciences\\
  Russia, Moscow, Gubkina Str. 8 \\
    and \\
    University of Science and Technology MISIS \\
    Russia, Moscow, Leninskiy prosp. 4 \\
    and \\
    Ivannikov Institute for System Programming of RAS\\
    Russia, Moscow, Alexander Solzhenitsyn Str., 25\\  
    \texttt{apechen@gmail.com}
    \And
      \orcid{0000-0001-8349-1738} Vsevolod Sakbaev \\
  Department of Mathematical Methods\\ for Quantum Technologies\\
  Steklov Mathematical Institute \\of Russian Academy of Sciences\\
  Russia, Moscow, Gubkina Str. 8 \\
  \texttt{fumi2003@mail.ru} 
}

\renewcommand{\shorttitle}{Quantum Random Walks and Quantum Oscillator in an Infinite-Dimensional Phase Space}

\begin{document}
\maketitle

\begin{abstract}
We consider quantum random walks in an infinite-dimensional phase space constructed using Weyl representation of the coordinate and momentum operators in the space of functions on a Hilbert space which are square integrable with respect to a shift-invariant measure. We study unitary groups of shift operators in the phase space and averaging of such shifts by Gaussian vectors, which form semigroups of self-adjoint contractions: we find conditions for their strong continuity and establish properties of their generators. Significant differences in their properties allow us to show the absence of the Fourier transform as a unitary transformation that implements the unitary equivalence of these compressive semigroups. Next, we prove the Taylor formula for a certain special subset of smooth functions for shifting to a non-finite vector. It  allows us to prove convergence of quantum random walks in the coordinate representation to the evolution of a diffusion process, as well as convergence of quantum random walks in both coordinate and momentum representations to the evolution semigroup of a quantum oscillator in an infinite-dimensional phase space. We find the special essential common domain of generators of semigroups arising in averaging of random shift operators both in position and momentum representations. The invariance of this common domain with respect to both semigroups allows to establish properties of a convex combination of both generators. That convex combination are Hamiltonians of infinite-dimensional quantum oscillators. Thus, we obtain that a Weyl representation of a random walk in an infinite dimensional phase space describes the semigroup of self-adjoint contractions whose generator is the Hamiltonian of an infinite dimensional harmonic oscillator.
\end{abstract}


\section{Introduction}

Random walks are a useful mathematical tool which found successful applications in physics, differential equations, computer science, biology, and other fields. Quantum random walks are constructed using compositions of non-commuting random operators. 
Quantum random walks were introduced in the seminal work~\cite{Aharonov_Davidovich_Zagury_1993} with applications in quantum optics. Since then, various mathematical results were obtained. 
In~\cite{Oksak} two explicitly solvable models of quantum random processes were described by the Langevin equation, namely, those for a $"$free$"$ quantum Brownian particle and for a quantum Brownian harmonic oscillator. The paper \cite{Attal_Dhahri_2010} is devoted to a description of a quantum noise model which is generated by the sequence of repeated quantum interactions with a chain of exterior systems. Analysis of this model is reduced to the analysis of compositions of mutually independent random unitary operators. The convergence of sequences of random walks is this and similar models were studied in \cite{Belton_2010,  Sahu_2008}. Random walks are applied for analysis of Anderson localization~\cite{Aizenman_Simone_2015,Chaturvedi_Srivastava_1983}. Quantum Langevin equation with quantum Poisson random process was derived in the low density limit~\cite{Accardi_Pechen_Volovich_2002,PechenJMP2004} closely related to the collisional decoherence~\cite{Filippov2020}. In \cite{Attal_Petruccione_Sabot_Sinayskiy_2012} models of open quantum random walks over graphs were studied as Markov chains of operators corresponding to different open quantum systems. This model of a quantum random walks on a graphs was studied in \cite{Dhamapurkar_Dahlsten_2023}, where it was shown that the node position statistics do not thermalize in the standard sense. In particular, quantum walks over fullerene graphs were shown to constitute a counter-example to the hypothesis that subsystems equilibrate to the Gibbs state~\cite{Kempe_2003, Kempe_2009}. Large time behavior of the quantum mechanical probability distribution of the position
observable in a finite-dimensional space when the sequence of unitary updates is given by an independent identically distributed (i.i.d) sequence of random matrices was studied in~\cite{Joye_2011, Joye_Merkli_2010, Joye_2012a}. Dependence of the properties of quantum random walks on the properties of the coordinate space was studied in~\cite{Wang_Wang_2020}. The analysis and comparison of quantum and classical discrete and continuous random walks were done in \cite{Venegas-Andraca_2012, Weickert}. Quantum multipole noise in indefinite metric pseudo-Hilbert spaces was shown to appear in corrections to quantum Brownian motion~\cite{Pechen_Volovich_2002}. In~\cite{Kolokoltsov_2020},  continuous time random walks were applied to modeling of the continuous quantum measurements yielding the new fractional in time evolution  quantum filtering equation and thus new fractional equation for open quantum systems. The ambiguity of quantization procedure was considered as a source of randomness of quantum Hamiltonian in~\cite{GOSS21}. Quantum random walks as compositions of i.i.d. random quantum dynamical semigroups were studied in~\cite{GOSS22}. 

Quantum random walks in infinite-dimensional spaces are studied in less details than quantum random walks in finite-dimensional spaces. One reason is the absence of the space for unitary representation of the group of shifts on an arbitrary vector in an infinite-dimensional space \cite{Vershik}. To obtain this representation we need in a shift-invariant measure on an infinite-dimensional phase space. According to the Weyl's theorem, there is no nonzero countably additive $\sigma$-finite locally finite measure on the $\sigma$-ring of Borel subsets of an infinite-dimensional topological vector space invariant with respect to shifts on vectors of this space~\cite{Weyl}. Various problems were studied for infinite-dimensional topological vector spaces about about existence of measures on invariant with respect to shifts to vectors from some maximum allowable subspace~\cite{Vershik}, shift-invariant measures which are not either $\sigma$-finite~\cite{Baker} or countably additive~\cite{S2016}, etc. However, invariant measure and corresponding  unitary representation of the group of shifts in infinite-dimensional space have not been constructed.

This work aims to fill this gap and aims towards developing rigorous theory for describing behavior of quantum random walks in infinite dimensional phase spaces. We call a non-negative additive function of a set defined on a certain ring of subsets of $E$ as non-negative measure on $E$. In \cite{S2016}, a family of shift-invariant, locally finite $\sigma$-finite, complete measures on $E$ is defined. We construct a non-negative measure on a Hilbert space that is invariant with respect to a group of shifts along an arbitrary homogeneous vector field and along a some wide class of Hamiltonian vector fields. This measure allows us to obtain a canonical Weyl  representation of commutation relation is the space of function that are square integrable with respect to invariant measure. We use this representation to construct quantum random walks in a phase space of an infinite dimensional Hamiltonian system by using the canonical Weyl representations of the group of shifts in the Hilbert space of functions on a position subspace of the phase space. The position subspace of the phase space is realized as a separable real Hilbert space $E$. The Hilbert space of Weyl representation consists of complex valued functions on the position space that are square integrable with respect to a shift-invariant measure. Since according to the A. Weil theorem~\cite{Weyl, Weil} there is no Lebesgue measure on an infinite dimensional Hilbert space, we introduce a special finitely-additive non-negative shift-invariant measure and develop the analysis of functions that are integrable with respect to this analog of the Lebesgue measure. For this, we study the space of complex-valued functions defined on $E$ that are square integrable with respect to some translation invariant measure on $E$. We show that in this space there is no unitary transformation that has all the properties of the Fourier transform. In Remark \ref{remark-fourier-can-exists}, however, it is indicated that such a unitary transformation can exist in the extension of this space.

A different approach is to consider shift-invariant generalized measures. The generalized measure is defined not as an additive function of the set, but as a linear functional on a suitable space of test functions. The use of a generalized shift-invariant measure allowed to define the Hilbert space of square integrable functions in which the unitary Fourier transform operates~\cite{SSh}.

The structure of this work is the following. In Sec.~\ref{Sec:Preliminary}, preliminary definitions and results are provided including adapted notions of random variables to vector- and operator-valued cases, properties of sesquilinear forms and symmetric operators. In Sec.~\ref{Sec:FiniteApproximation}, following \cite{OSZ-20} an evolution  
of a finite-dimensional quantum oscillator is approximated by a composition of random shifts in the coordinate and momentum representations. In Sec.~\ref{Sec:Shifts} we define a finite additive translation invariant measure on $E$ and the space of square integrable functions with respect to it, and study properties of the operators acting by shifting argument of these functions in the coordinate and momentum representations. In Sec.~\ref{Sec:Averaging} we study general properties of semigroups arising as averaging of operators of a random shift in coordinate and momentum representations. In Sec.~\ref{Sec:Generators} we study properties of generators of these semigroups. The key result here is that positive combination of such generators under some conditions is self-adjoint and generates a strongly continuous semigroup. This allows us to define a Hamiltonian of an infinite-dimensional quantum oscillator. In Sec.~\ref{Sec:Fourier} we prove the absence of the Fourier transform on the space of square integrable with respect to shift-invariant measure functions on E. However, partial Fourier transforms with respect to a finite number of variables are introduced and their properties are analysed. In Sec.~\ref{Sec:Taylor} we prove Taylor's formula for shift of argument of smooth functions (def. \ref{def-smooth-functions}) along a vector from a certain subspace in coordinate and momentum representations. We also study smoothness of the image of a function under action of operators of the averaging of random shifts. In Sec.~\ref{Sec:InfiniteApproximation} we prove that evolution of a diffusion process on $E$ and evolution of an infinite-dimensional quantum oscillator on $E$ can be approximated by a composition of random shifts in the coordinate and momentum representations. Thus, we generalize the results of section \ref{Sec:FiniteApproximation} to the infinite-dimensional case. Conclusions Sec.~\ref{Sec:Conclusions} summaries the results and discusses open problems. 

\section{Preliminary information}\label{Sec:Preliminary}

In this section we introduce some definitions and technical results. Particularly, we provide notions of random variables adapted to vector- and operator-valued cases in infinite-dimensional spaces and Chernoff theorem. We remind formal definition of Weyl representation. We also discuss a connections between densely defined sesquilinear forms and symmetric operators. Using the representation theorem, we show that if one sesquilinear form is bounded by another then the domains of the operators which represent these forms are nested. 

Let $H$ be a complex Hilbert space (possibly non-separable). Let $B(H)$ denote the Banach space of bounded linear operators on $H$ endowed with the operator norm. Let $Y_{\rm s}(\mathbb{R}_+, B(H))$ be the set of strongly continuous functions from $\mathbb{R}_+=[0,+\infty)$ to the space $B(H)$ of bounded operators on $H$. Topology $\tau_{\rm s}$ on $Y_{\rm s}(\mathbb{R}_+, B(H))$ is defined by the family of functionals $\phi_{T,u}$ where $u \in H$ and $T>0$, acting as
$$
\phi_{T,u} (F) = \sup\limits_{t \in [0,t]} \|F(t) u\|_{B(H)}.
$$

\subsection{Measures and functional spaces}

In \cite{S2016, SITO}, a finitely additive measure $\lambda$ was introduced on $E$, defined on the minimal ring ${\rm r}$ containing absolutely measurable blocks (so called "infinite-dimensional rectangles" whose products of edge lengths converge absolutely), which is invariant with respect to any shifts and orthogonal transformations of $E$. For each orthonormal basis ${\mathcal E}=\{e_j\}$ of $E$, the restriction $\lambda _{\mathcal E}$ of  measure $\lambda$ onto the minimal ring ${\rm r}_{\mathcal E}$ containing measurable blocks with edges, co-directional with the vectors of the basis $\mathcal E$, was investigated. Using the standard completion procedure, the measure $\lambda$ was expanded from the ring $\rm r$ to the ring $\lambda$-measurable sets ${\cal R}$, and the measure $\lambda_{\mathcal E}$ to the ring $\lambda_{\mathcal E}$-measurable sets ${\cal R}_{\mathcal E}$. The Hilbert space ${\mathcal H}=L_2(E,{\cal R} ,\lambda ,\mathbb{C})$ was introduced as equivalence classes of complex-valued functions $u: E \to \mathbb{C}$ that are square integrable with respect to $\lambda$ and, similarly, for each orthonormal basis (ONB) $\cal E$ of the space $E$ the Hilbert space 
${\mathcal H}_{\mathcal E}=L_2(E,{\cal R}_{\mathcal E},\lambda _{\mathcal E},\mathbb{C})$ was introduced. It was determined
that for any orthonormal basis $\mathcal E$, the space ${\mathcal H}_{\mathcal E}$ is not separable and is a closed subspace in ${\mathcal H}$. 

In \cite{S2016}, a unitary representation in the Hilbert space ${\cal H}_{\cal E}$ was obtained for
the Abelian group of vectors of the space $E$ by the unitary Abelian group ${\bf S}_{h}, \, h\in E,$ of shift operators on $E$. A diffusive semigroup of self-adjoint contractions ${\bf U}_{\gamma_t},\, t\geq 0$ was investigated, which arises when averaging a Gaussian random process given by a family of centered Gaussian measures $\gamma _t,\, t\geq 0,$ on the space $E$ forming a semigroup regarding the convolution operation $\gamma _t*\gamma _s=\gamma _{t+s},\, t,s\geq 0$. The criterion of strong continuity of the one-parameter shift group ${\bf S}_{th},\, t\in\mathbb R$, along a constant vector field and the criterion of strong continuity of the diffusion semigroup ${\bf U}_{\gamma_t},\, t\geq 0,$ in the space ${\cal H}_{\cal E}$ was obtained.

A Banach space of bounded linear operators in a Banach space $X$ is denoted by $B(X)$.

\subsection{Weyl representation of canonical commutation relations}

The phase space $F=E\oplus E$, endowed with the Euclidean norm, is an Abelian Banach topological group with respect to the addition of elements. For each ONB ${\cal E}=\{ e_k\}$ in the space $E$ the symbol $L_1({\cal E})$ denotes the subspace $\{ x\in E:\ \{ (x,e_k)\} \in \ell _1\}$ equipped with the new norm $\| x\|_{L_1({\cal E})}=\| \{ (x,e_k)\}\|_{\ell _1}$. Then, the subspace $F_1({\cal E})=L_1({\cal E}) \oplus L_1({\cal E})$ is a densely embedded in the Banach group $F$ subgroup that is Banach with respect to the new norm $\| \cdot \|_{L_1}$.

A dense on $F$ map $F_1({\cal E})\to B(H_{\cal E}(E))$ which associates to any vector $(h,a)\in F_1({\cal E})$ the unitary operator ${\bf W}_{h,a} \in B(H_{\cal E})$ is called a unitary Weyl representation of the Abelian group $F_1({\cal E})$ in the space $H_{\cal E}$ (\cite{RS2, H}) if the following conditions are satisfied
\begin{enumerate}[label=(\roman*)]
     \item The operator ${\bf W}(z)$ is unitary for each $z \in F_1({\cal E})$,
     \item For any $z_1,z_2 \in F_1$ the equality
     $$
      \qquad {\bf W}_{z_1}{\bf W}_{z_2}=\exp(i\omega (z_1,z_2))W_{z_1+z_2}
     $$
     holds, where $\omega $ is a symplectic form on $F$, 
     \item Convergence $\| z_n-z_0\|_{F_1}\to 0$ 
     as $n\to \infty$ implies  $\| ({\bf W}(z_n)-{\bf W}(z_0))f\|_{H_{\cal E}}\to 0$.
\end{enumerate}
\begin{align}
{\bf W}_{h_1,a_1}{\bf W}_{h_2,a_2}=\exp \Bigl\{-\frac{i}{2}\Bigl[(x_1,a_2)_E-(x_2,a_1)_E\Bigr]\Bigr\}{\bf W}_{h_1+h_2,a_1+a_2},\nonumber\\ (h_1,a_1), (h_2,a_2)
\in L_1 ({\cal E})\oplus L_1({\cal E}).\label{0}
\end{align}
CCR representation
(\ref{0})
can be realized by the mapping $(h,a)\to {\bf W}_{h,a}$ defined by the equality
\begin{equation}\label{1}
{\bf W}_{h,a}=\exp (\frac{i}{2}(h,a)_E){\bf S}_h{\bf V}_a,\ (h,a)\in L_1 ({\cal E}){\oplus} L_1({\cal E}).
\end{equation}
(see \cite{H}).
Note that the family of operators that implement canonical commutation relations (\ref{0}) generates an algebra of operators describing bosons~\cite{Bogo}, Chapter 6.4).

In \cite{OSZ-20}, in the case of a finite-dimensional Euclidean space $E$, random walks in the momentum space of a quantum system were investigated, given by compositions of independent identically distributed random multiplication operators ${\bf M}_h$ by an imaginary exponent $\exp (i(x,h))$ from the linear function $(\cdot ,h)$ on the coordinate space $E$, where $h$ is a Gaussian random variable with values in the space $E$. The mathematical expectation of shift operators on a random vector $h$ in the momentum space according to the Gaussian measure $\gamma _t$ on the space $E$ with a non-negative covariance operator $t{\bf D},\, t\geq 0,$ is a semigroup ${\bf A}_{\bf D}(t),\ t\geq 0,$ of multiplication operators by the function $\exp(-t({\bf D}x,x)),\, t\geq 0$~\cite{OSZ-20}. In this case, the Fourier transform on the space $L_2(E)$ implements the unitary equivalence of the shift operator ${\bf S}_h$ along the vector $h$ in coordinate space and the operator ${\bf M}_h$ shift along the same vector in momentum space, as well as the unitary equivalence of the semigroup ${\bf U}_{\gamma _t},\, t\geq 0,$ diffusion operators and semigroups ${\bf A}_{\bf D}(t),\, t\geq 0.$

In this paper, in the case of an infinite-dimensional Euclidean space $E$, the properties of the strong continuity of the semigroup ${\bf A}_{\bf D}(t),\,t\geq 0,$ of multiplication operators by the function $\exp(-t({\bf D}x,x)),\, t\geq 0$ are investigated as soon as properties of semigroups ${\bf U}_{\gamma _t},\, t\geq 0,$ of convolution operators with centered Gaussian measures $\gamma_{t{\bf D}},\, t\geq 0,$ whose images under Fourier transform are the functions $\exp(-t({\bf D}x,x)),\, t\geq 0$. The generator ${\bf W}_{\bf D}$ of the semigroup of self-adjoint contractions ${\bf A}_{\bf D}$ is investigated by the same methods as the generator of the semigroup ${\bf U}_{\gamma}$ is investigated in \cite{SITO}. It is shown that in the case of a separable real infinite-dimensional space $E$, there is no such unitary transformation of the space ${\cal H}_{\cal E}$ onto itself, which, with an arbitrary choice of the non-negative nuclear operator $\bf D$, implements the unitary equivalence of the semigroup ${\bf A}_{\bf D}(t),\, t\geq 0,$
and semigroups ${\bf U}_{\gamma _t},\, t\geq 0.$

\subsection{Random operator-valued variables}

Let $Y_{\rm s}(\mathbb{R}_+,B(H))$ be a topological vector space of strongly continuous mappings ${\mathbb R}_+\to B(H)$. 
The topology on the space  $Y_{\rm s}(\mathbb{R}_+,B(H))$ is defined by the family of functionals $\{ {\phi }_{t,u},\ t\in {\mathbb R}_+,\, u\in H\}$ where $\phi _{t,u}(F)=\| F(t)u\|_H,\ F\in  Y_{\rm s}(\mathbb{R}_+,B(H))$.

If $(\Omega,\Sigma,P)$ is some probability space, then we call a measurable map $F:\ \Omega\to Y_{\rm s}(\mathbb{R}_+,B(H))$ as a random operator-valued function. It means that $(F(t)u,v)$ is a complex valued random variable for every $t\in {\mathbb R}_+$ and any $u,v\in H$.  If its image is almost surely a semigroup (a unitary semigroup), then we call it a random semigroup
(a random unitary semigroup).

Since the space $H$ is separable, the sum and a product of any to random operator valued function $F,G$ are random operator values function, $F(t)u,\ t\in {\mathbb R}_+$ is a random vector valued function for every $u\in H$ and $\|F(t)\|_{B(H)},\ t\in {\mathbb R}_+$ is a random real valued function.

The expectation of a random operator $A:\ \Omega\to B(H)$ is defined in the sense of the Pettis integral:
$$
(x, \mathbb{E}(A_\omega)y) = \int\limits_\Omega (x, A_\omega y) dP(\omega) \quad \forall x,y \in H.
$$
The variance is defined as
$$
\mathbb{D}(A_\omega) = \mathbb{E} [(A - \mathbb{E} A_\omega)^*(A - \mathbb{E} A_\omega)].
$$
The independence of the random operators $A$ and $B$ is defined as the independence of $(x_1, A x_2)$, $(x_3, B x_4)$ for all vectors $x_1, x_2, x_3, x_4$.

For each random semigroup $U_\omega(t): \Omega \to Y_{\rm s}(\mathbb{R}_+, B(H))$ and $n\in {\mathbb Z}_+$ we define its averaging $\Av_n U(t): \Omega^n\to Y_{\rm s}(\mathbb{R}_+, B(H))$ as
$$
[\Av_n U]_{\omega_1, \ldots, \omega_n}(t) = U_{\omega_1}(t/n) \circ U_{\omega_2}(t/n) \circ \ldots \circ U_{\omega_n}(t/n).
$$ 
\begin{lemma}\label{expect-independent-unitary}(proved in~\cite{SS18feynman})
    Let $U_1, \ldots, U_n$ be independent random unitary operators. Then
    $$\mathbb{E} (U_1 \circ U_2 \circ \ldots \circ U_n) 
    = \mathbb{E} U_1 \circ \mathbb{E} U_2 \circ \ldots \circ \mathbb{E} U_n.$$
\end{lemma}

The convolution of a measurable $n$-dimensional function with the measure $P$ is denoted as
$$
f(x) \star P := \int f(x-t) d P(t).
$$
In particular, if $P$ is defined by density $g(x)$, the expression above coincides with the convolution of functions: 
$$
f(x) \star g(x) := \int f(x-t)g(t) dt = \int f(t)g(x-t) dt.
$$
If $T(t), \, t\ge 0$ is a strongly continuous operator-valued function, then we call the function $(T(t/n))^n$ as $n$-th Chernov iteration. The following classical theorem allows to find the limit of such iterations at $n\to\infty$.

\begin{theorem}[Chernoff, \cite{engel-nagel}, Corollary 5.3 chapter III]
    Let function $V: \mathbb{R}_+ \to B(H)$ satisfy the conditions:
    \begin{enumerate}
        \item $V(0) = I$; 
        \item $\| V(t)^k \| \le M e^{wkt}$ for some $M,k\in\mathbb{R}$ and all $k\in\mathbb{N}, t\in\mathbb{R}_+$;   
        \item There exist a linear subspace $D\subset H$ and the operator $A$ on $D$ such that for all $x\in D$ 
        $$
        Ax = \lim\limits_{t\to 0} \dfrac{1}{t} (V(t)x - x);
        $$
        \item Either there exists such $\lambda_0 > w$ that both $D$ and $(\lambda_0 - A)D$ are dense in $H$, or the closure of $\bar{A}$ of the operator $A$ is a generator of some strongly continuous one-parameter semigroup.
    \end{enumerate}
    Then the closure $\bar{A}$ of the operator $A$ generates a strongly continuous semigroup $T(t): \mathbb{R}_+\to B(H)$ such that
    $$ 
    \lim\limits_{k\to\infty} \sup\limits_{t\in[0,T]} \|T(t) x - (V(t/k))^k x\|_H = 0
    $$
    for all $x\in H$ and the following inequality is satisfied
    $$
    \|T(t)\|_{B(H)} \le M e^{wt}.
    $$
\end{theorem}

We also use a theorem which, by the existence of the first $k$ moments of the $n$-dimensional distribution, guarantees the existence of the first $k$ derivatives of the characteristic function. The special case for one-dimensional distributions is described in theorem 1 of paragraph 2.12 \cite{shiryaev}. Below we give the proof of the general case which is close to the proof in~\cite{shiryaev}.
\begin{lemma}\label{lemma-momentum-charf}
    If for some random vector $X$ all moments
    $$
    m_{j_1, \ldots, j_m} = \mathbb{E} X_{j_1} \ldots X_{j_m}
    $$
    up to order $n$ are finite, then for its characteristic function $\phi(t) = \int e^{i(t,x)} dP(x)$ all mixed derivatives up to order  $n$ are  well defined, and
    \begin{equation}\label{eq-charf-der}
    \dfrac{\d^m}{\d t_{j_1} \ldots \d t_{jm}} \phi(t) = \int\limits_{\mathbb{R}^n} (i x_{j_1}) \ldots(i x_{j_m}) e^{i(t,x)} P(x), \quad m \le n.
    \end{equation}
    Thus,
    \begin{equation}\label{eq-momentum-by-charf}
    m_{j_1, \ldots, j_m} = i^m \dfrac{\d^m}{\d t_{j_1} \ldots \d t_{j_m}} \phi(0).
    \end{equation}
\end{lemma}

\begin{proof}
    Equality \ref{eq-momentum-by-charf} follows directly from \ref{eq-charf-der}. 
    
    We prove  \ref{eq-momentum-by-charf} by induction.
    
    According to the definition of derivative, we have
    $$
     \dfrac{\d}{\d t_{j_{m+1}}} 
         \dfrac{\d^{m}}{\d t_{j_1} \ldots \d t_{j_m}} \phi(t) 
         = \lim\limits_{\delta \to 0} \dfrac{1}{\delta} \left(
            \dfrac{\d^{m}}{d t_{j_1} \ldots \d t_{j_m}} \phi(t + \delta e_{j_{m+1}}) 
            - \dfrac{\d^{m}}{d t_{j_1} \ldots \d t_{j_m}} \phi(t)
         \right),
     $$
     where according to the base of induction,
     $$\dfrac{\d^{m}}{\d t_{j_1} \ldots \d t_{j_m}} \phi(t) =  \int\limits_{\mathbb{R}^n} (i x_{j_1}) \ldots (i x_{j_m}) e^{i(t,x)} P(x).
    $$
    Substituting, we get
    $$
     \dfrac{\d}{\d t_{j_{m+1}}} 
         \dfrac{\d^{m}}{\d t_{j_1} \ldots \d t_{j_m}} \phi(t) 
         =  \lim\limits_{\delta \to 0}   \int\limits_{\mathbb{R}^n} \dfrac{1}{\delta} (i x_{j_1}) \ldots (i x_{j_m}) e^{i(t,x)} (e^{i\delta x_{j_{m+1}}} - 1) P(x).
    $$
    In the last expression, one can rearrange the limit with an integral according to Lebesgue's dominated convergence theorem. Indeed, the expression under the integral sign can be evaluated uniformly over $\delta$ as
    $$
    \left| \dfrac{1}{\delta} (i x_{j_1}) \ldots (i x_{j_m}) e^{i(t,x)} (e^{i\delta x_{j_{m+1}}} - 1)\right| \le 
    \dfrac{1}{\delta} |x_{j_1}| \ldots |x_{j_m}| \cdot |\delta x_{j_{m+1}}|
    \le |x_{j_1}| \ldots |x_{j_{m+1}}|
    $$
    since the inequality $|e^{ia} - 1| \le a$ is true for any real $a$.
    
    Considering that
    $$
    \lim\limits_{\delta \to 0} \dfrac{1}{\delta} (e^{i\delta x_{j_{m+1}}} - 1) = i x_{j_{m+1}}
    $$
    we obtain the required equality.
\end{proof}

The following lemma is a consequence of statement 12.2.8 of
\cite{BogachevSmolyanov}

\begin{lemma}(mean-value theorem)
    If $f \in C^1(\mathbb{R}_+, H)$, where $H$ is a complex or real Hilbert space (not necessarily separable), then
    $$
    \|f(b) - f(a)\| \le (b-a) \sup\limits_{t \in [a,b]} \|f'(t)\|.
    $$
\end{lemma}

The next lemma is a generalization of the classical result from real analysis.

\begin{lemma}\label{lemma-uniform-der-conv}
Let $f_n(t), A(t):\mathbb{R}_+\to H$, where $H$ is a real or complex Hilbert space, and all $A_n(t)$ are continuously differentiable in $t$. Let $A_n(t)$ converge to $A(t)$ pointwise, and assume that for some function $B:\mathbb {R}_+\to(\mathcal{H})$ the convergence holds:
    $$
    \forall T > 0: \quad \lim\limits_{m\to \infty} \sup\limits_{t \in [0,T]} \| A_m'(t) - B(t)\| = 0.
    $$
Then $A(t)$ is continuously differentiable, $A'(t) = B(t)$, and
    $$
    \forall T > 0: \quad \lim\limits_{m\to \infty} \sup\limits_{t \in [0,T]} \| A_m(t) - A(t)\| = 0.
    $$
\end{lemma}
\begin{proof}
Fix $T > 0$, select $a \in [0,T]$ and $x\in [0,T] \setminus\{a\}$. Consider the auxiliary function
$$
F(s) = A_n(a + s(x-a)) - s A'_n(a)(x-a), \quad s \in [0,1].
$$
Its derivative is
$$
F'(s) = (x-a) A'_n(a+s(x-a)) - A'_n(a)(x-a).
$$
According to the mean-value theorem,
    $$
    \|F(1) - F(0)\| \le \sup\limits_{s\in [0,1]} \|F'(s)\| 
    = (x-a) \sup\limits_{s\in [0,1]} \|A'_n(a) - A'_n(a + s(x-a))\|.
    $$
Now apply the triangle inequality:
\begin{align*}
\|F(1) - F(0)\|\le& (x-a) \sup\limits_{s\in [0,1]}  
\|A'_n(a) - B(a)\|\\
&+ \sup\limits_{s\in [0,1]}
        \|B(a) - B(a + s(x-a))\|\\
&+\sup\limits_{s\in [0,1]}\|B(a + s(x-a)) - A'_n(a + s(x-a))\|.
\end{align*}
Now perform the pointwise limit as $n\to\infty$ and substitute the values $F(0)$ and $F(1)$:
    $$
    \|A(x) - A(a) - B(a)(x-a)\| \le (x-a) \sup\limits_{s\in [0,1]}   \|B(a) - B(a + s(x-a))\|.
    $$
    Since $B(s)$ is a uniform limit of continuous functions, it is continuous and
    $$
    \lim\limits_{x \to a} \sup\limits_{s\in [0,1]}   \|B(a) - B(a + s(x-a))\| = 0.
    $$
    Thus, 
    $$
    \lim\limits_{x \to a} \|A(x) - A(a) - B(a)(x-a)\| = 0
    $$
     and $A'(a) = B(a)$.

Now let us prove that $A_n(t)$ converges to $A(t)$ uniformly over the segments.

Consider the auxiliary function
     $$
     G(s) = A(a + s(x-a)) - A_n(a + s(x-a)).
     $$
     It is differentiable and
     $$
     G'(s) = (x-a) A'(a+s(x-a)) - (x-a) A'_n(a+s(x-a)).
     $$
     Using the mean-value theorem again, we get
     $$
     \|G(1) - G(0)\| \le \sup\limits_{s\in [0,1]} \|G'(s)\|.
     $$
     Therefore
\begin{align*}
\|A(x) - A_n(x)\| &\le \| A(x) - A_n(x) - (A(a) - A_n(0))\| + \|A(a) - A_n(0))\| \\
&=\|G(1) - G(0)\| + \|A(a) - A_n(0))\|\\
&=|x-a| \sup\limits_{s\in[0,1]} \|A'(a+s(x-a)) - (x-a) A'_n(a+s(x-a))\| + \|A(a) - A_n(0))\|.
\end{align*}
Since $A'_n$ converges to $A_n$ uniformly over the segments, the norm $\|A(x) - A_n(x)\|$ converges to zero uniformly over $x$.
\end{proof}

\subsection{Sesquilinear forms}
A sesquilinear bilinear form (for more details, see, for example, Chapter VI of \cite{kato72}) $t:H\times H\to\mathbb{C}$ over a complex Hilbert space $H$ is a map defined on the linear subspace $D=D(t)\subset H$, which is linear in the first variable and semi-linear in the second variable, that is, for any $x,y\in D$ the form $t(x,y)$ is defined, and for any $\alpha,\beta\in\mathbb{C}, x,y,z\in D$ one has
\begin{align*}
    t(\alpha x + \beta y, z) &= \alpha t(x,z) +\beta t(y,z), \\
    t(x, \alpha y + \beta z) &= \bar{\alpha} t(x,y) + \bar{\beta} t(x,z).
\end{align*}
If $D(t) = H$, then the form is bounded, i.e. $|t(x,y)|\le C\|x\|_H\|y\|_H$ for some $C> 0$. A form is called densely defined if $D$ is dense in $H$. We assume that all considered in this work forms are densely defined, unless explicitly stated otherwise. The form is called symmetric if $t(x,y) = \overline{t(y,x)}$ for all $x,y\in D$. We denote the quadratic form $t[x] = t(x,x)$ by the same symbol as the bilinear form. Since in complex case the bilinear form is uniquely restored by the quadratic one via polarization identity, this should not lead to a confusion.

Two forms are called equal if their have the same domain and equal values on this domain. The form $t_2$ is called an extension of $t_1$ if $D(t_1)\subset D(t_2)$, and the restriction of $t_2$ to $D(t_1)$ coincides with $t_1$. In this case, we write $t_1 \subset t_2$.

The form is called bounded from below if $t[x]\ge\gamma\|x\|_H$, which is written as $t\ge\gamma$. If $t \ge 0$, then the form is called positive-definite. If by the form $t + \gamma, \gamma \in \mathbb{R}$ we mean the form acting according to the rule $(t+\gamma)(x,y) = t(x,y) + \gamma(x,y)_H$, then $t\ge \gamma$ is equivalent to the positive definiteness of the form $t - \gamma$.

The sequence $\{x_n\}\subset H$ is called $t$-converging to $x\in H$ if $\{x_n\}\subset D(t)$, $t_n\to t$ and $t[x_n - x_m]\to 0$ when $n,m \to \infty$. The form $t$ is called closed if $x_n \xrightarrow{t}x$ entails $x\in D(t)$ and $t[x_n - x] \to 0 $.
A form is called closable if it has a closed extension. In this case, among its closed extensions there is a minimal  by inclusion form, which is called the closure of the form $t$ and denoted $\bar t$.
If the symmetric form $t$ is defined as $t(x,y) = (Ax,y)$, where $A$ is symmetric, and $D(t) = D(A)$, then $t$ is
closable (\cite{kato72}, chapter VI, vol. 1.27).

The representation theorem (\cite{kato72}, sec. 2.1 and sec. 2.6) guarantees that for any closed symmetric positively and densely defined form $t$ there exists a self-adjoint operator $T$ such that $D(T)\subset D(t)$, and for all $x \in D(T), y \in D(t)$ the condition $t(x,y) = (Tx,y)_H$ is satisfied. Moreover, $D(T)$ is an essential domain of $t$, that is, $t$ is the closure of its restriction to $D(T)$.

Let $T$ be a symmetric densely defined operator, and $T\ge\gamma$. Then the form $(t -\gamma)(x,y) =(Tx,y) -\gamma(x,y)_H$ is closed and positive-definite, and there exists a self-adjoint operator $T_\gamma$ representing the form $t -\gamma$. The operator $T_\gamma +\gamma$ is self-adjoint and does not depend on the choice of $\gamma$. It is called the Friedrichs extension of the operator $T$ and is denoted by $T_{\rm f}$.

We write $t_1 \le t_2$ if $D(t_2)\subset D(t_1)$, and $t_2 - t_1$ is positively defined.

If $t(x,y) =(Sx,Sy), D(t)=D(S)$ for some operator $S$, then $t$ is always non-negative and symmetric, and $t$ is closed if and only if the operator $S$ is closed (\cite{kato72}, chapter VI, examples 1.3 and 1.13).

\begin{theorem}\label{th-repr-operators-incl}
    Consider two closed densily defined positive defined symmetric forms $t_1$ and $t_2$ such that for some positive constant $C$ the inequality $t_1\le C t_2$ holds. Let $A_1$ and $A_2$  be the operators representing $t_1$ and $t_2$ respectively. Then $D(A_2) \subset D(A_1)$.
\end{theorem}
\begin{proof}
    Since $t_1$ and $t_2$ are symmetric, closed, and positive-definite, we can use the representation theorem, and the operators $A_1$ and $A_2$ are defined correctly and uniquely.
    In addition, the operators $A_1$ and $A_2$ are self-adjoint and positively defined.

    Using the spectral theorem, represent both operators as
    $$
    A_i = \int\limits_\mathbb{R} \lambda d E_i(\lambda),
    $$
    where $d E_i(\lambda)$ is a spectral measure (\cite{kato72} \S 6.5).

    The domain of each operator is the set of all $x\in H$ such that
    $$
    \int\limits_\mathbb{R} \lambda^2 d (E_i(\lambda) x, x) < \infty .
    $$
    Since the operators $A_i$ are positive-definite, then $E_i(0) = 0$, and for any $x\in D(A_i)$ the following operators are well defined
    $$
    A_i^{1/2} = \int\limits_\mathbb{R} \sqrt{\lambda} d E_i(\lambda) .
    $$
    Both operators $A_i^{1/2}$ are symmetric by definition. The domains of their closures coincide with the domains of the forms $t_i$. Indeed, $x \in D(t_i)$ if there exists a sequence $\{x_n\} \in (A_i)$ such that $x_k \xrightarrow{t} x$, that is
    \begin{equation}\label{t-conv}
        \|x_n - x\|_H \to 0,\quad t_i(x_n - x_m) \to 0, \quad n,m \to 0.
    \end{equation}
    On the other hand, $x \in D(\bar A_i^{1/2})$ if there is a sequence $\{x_n\}\inf(A_i)$ such that
    \begin{equation}\label{A-conv}
        \|x_n - x\|_H \to 0,\quad \|A^{1/2}_i(x_n - x_m)\|_H \to 0, \quad n,m \to 0.
    \end{equation}
    At the same time, it is true for $x_n, x_m \in D(A_i)$ that 
    $$
    t_i(x_n - x_m) 
    = (A_i(x_n - x_m), x_n - x_m)_H 
    = \|A^{1/2}_i(x_n - x_m)\|_H
    $$
    which gives the equivalence of the conditions \ref{t-conv} and \ref{A-conv}.

    If we consider an arbitrary element $x\in D(A_2)$ and an arbitrary segment $[a, b]$, then $y= E_2(b) x - E_2(a)x$ both lie in the domain  of $A_2$ and in the domain of the forms $t_i$, that is, in the domain
    of the operators $\bar A_i^{1/2}$. Additionally, the following equality holds:
    $$
    t_1[y] 
    = \int\limits_a^b \lambda d (E_1(\lambda) x, x) 
    \le C \int\limits_a^b \lambda d (E_2(\lambda) x, x)
    = C t_2[y].
    $$
    Arbitrariness of $a$ and $b$ implies that the measure $d\mu(x) = d C(E_2(\lambda)x,x) - d(E_1(\lambda)x, x)$ is positively defined. Therefore,
    $$
    \int\limits_\mathbb{R} \lambda^2 d (E_1(\lambda) x,x) 
    = C \int\limits_\mathbb{R} \lambda^2 d (E_2(\lambda) x,x) 
    - \int\limits_\mathbb{R} \lambda^2 d \mu(x) < + \infty
    $$
    and, by definition, $x\in D(A_1)$.
\end{proof}

\section{Approximation of the evolution of a finite-dimensional quantum oscillator with random walks}
\label{Sec:FiniteApproximation}

Before proceeding to the infinite-dimensional case, we present, for the convenience of the reader, the finite-dimensional results published in \cite{OSZ-20}. We show how evolution of a diffusion process and evolution of a quantum oscillator can be approximated by a composition of random shifts in the coordinate and momentum representations.

\begin{definition}\label{def-sol}
    We say that the semigroup $\mathfrak{U}(t): L_2(\mathbb{R}^n)\to L_2(\mathbb{R}^n)$ represents the solution of the partial differential equation
    $$
    \dfrac{\d}{\d t} u(t,x) = D u(t,x), \quad t \ge 0, x \in \mathbb{R}^n,
    $$
    where $D$ is some differential operator, with the initial condition
    $$
    u(0,x) = u_0(x)
    $$
    if the function
    $$
    u(t,x) = \mathfrak{U}(t) u_0(x)
    $$
    is its solution.
\end{definition}

\begin{definition}
    Denote by $S_{th}: L_2(\mathbb{R}^n)\to L_2(\mathbb{R}^n),\, t\ge 0, h\in\mathbb{R}^n$ the unitary group of shift of the function argument along vector $h$
    $$
    S_{th} f(x) = f(x - th)
    $$
    and denote by $\widehat S_{th}$ its Fourier transform, which we call a shift by the vector $h$ in the momentum representation:
    $$
    \widehat S_{th} g(y) = e^{it(h,y)} g(y).
    $$
\end{definition}

The generator of $S_{th}$ is the differentiation operator along the direction $h$, and the generator of $\widehat S_{th}$ is the operator of multiplication by the function $i(h,y)$. Both of these operators are unbounded, but according to Stone's theorem they are defined on dense subspaces and are the product of self-adjoint operators and a complex unit.

\begin{lemma}\label{lemma-finite-der-at-zero}
Let $h: \Omega\to \mathbb{R}^n$ be a random vector with zero expectation and variance $D$. Then the operator-valued functions $\mathbb{E}_h S_{\sqrt{t}h}$ and $\mathbb{E}_h\widehat S_{\sqrt{t}h}$ are differentiable at zero, and for any $f \in L_2(\mathbb{R}^n)$
\begin{align*}
\left.\left( \dfrac{\d}{\d t} \mathbb{E}_h S_{\sqrt{t}h} f(x) \right)\right|_{t=0}&=
\dfrac{1}{2} \Delta_D f(x)= \dfrac{1}{2} \sum\limits_{i,j=1}^n \dfrac{d^2}{dx_i d x_j} f(x),\\
\left.\left( \dfrac{\d}{\d t}  \mathbb{E}_h \widehat S_{\sqrt{t}h} f(x) \right)\right|_{t=0}&= - \dfrac{1}{2} (Dx,x) f(x)= - \dfrac{1}{2} \sum\limits_{i,j=1}^n d_{i,j} x_i x_j \cdot f(x).
\end{align*}
\end{lemma}

\begin{proof}
    First consider averaging the shift in the momentum representation:
    $$
    \mathbb{E}_h \widehat S_{\sqrt{t}h} f(x) 
    = \int\limits_\Omega e^{i\sqrt{t}(h,x)} f(x) dP(h)
    = f(x) \phi_h(\sqrt{t} x),
    $$
    where $\phi_h(x)$ is the characteristic function of the random vector $h$.

    According to the assumption, the first two moments of the distribution $h$ are finite and equal to $0$ and $D$ respectively. By lemma \ref{lemma-momentum-charf} we can easily find the gradient and Hessian of the characteristic function at zero. Let us decompose the characteristic function of the random vector $h$ according to the Taylor formula:
    $$
    \phi_h(x) = 1 - \dfrac{1}{2}(Dx,x) + o(|x|^2), \quad x\to 0.
    $$
    For a fixed $x$ and $t\to 0$,  
    $$
    \phi_h(x\sqrt{t}) = 1 - \dfrac{1}{2} t (Dx,x) + o(t).
    $$
    Therefore, the derivative at the point $t=0$ of the operator-valued function of multiplication by $\phi_h(x\sqrt{t})$ is the operator of multiplication by $-\dfrac{1}{2}(Dx,x)$.

    The Fourier transform translates a shift in the momentum representation into a shift in the coordinate representation, and $h$ into a random vector $\widehat h$, whose expectation and variance equal to $0$ and $D$. Therefore, the function $\mathbb{E}_h S_{\sqrt{t}h}f(x)$ is unitary equivalent to the function $\mathbb{E}_{\widehat h}\widehat S_{\sqrt{t}h} \widehat f(x)$, and its derivative at zero equals
    $$ 
    \widehat{-\dfrac{1}{2} (Dx,x)} 
    = \dfrac{1}{2} \Delta_D.
    $$
\end{proof}

Closure of the operator of multiplication by $-(Dx,x)$ is the generator of the semigroup $\widehat{U}_D(t)$, acting according to the rule
$$
\widehat{U}_D(t) f(x) = \exp \left(-t (Dx,x)\right) f(x).
$$
Therefore, the same is true for the closure of the operator $\Delta_D$: it is generator of some semigroup $U_D(t)$. By applying the Fourier transform, we can find what kind of semigroup it is:
$$
U_D(t) f(x) 
= \widehat{\exp \left(-t (Dx,x)\right)} \star f(x) = \int\limits_{\mathbb{R}^n} f(x-h) d \mu_{tD}(h),
$$
where $\mu_{tD}(h)$ is the measure whose characteristic function has the form $\exp\left(-t(Dx,x)\right)$, that is, a Gaussian measure with zero expectation and variance $tD$. By the definition of the generator, for all functions $f$ from some dense subspace $L_2(\mathbb{R}^n)$ it is true that
$$
\left.\left(\dfrac{\d}{\d t} U_D(t) f(x) \right)\right|_{t=0} 
=  \Delta_D f(x).
$$
So the semigroup $U_D(t)$ represents the solution of the Laplace equation with the boundary condition $u_0= f(x)$ in the sense of the definition \ref{def-sol}. At the same time, due to the semigroup property
$$
\Av_m U_D(t) = (U_D(t/m))^m = U_D(t).
$$
Therefore $\Av_m U_D(t)$ also represents the solution of the Laplace equation. The following theorem shows that if we take an arbitrary probability measure with the same first and second moments instead of a Gaussian distribution, then the limit of such averaging at $m\to\infty$ in a certain topology converges to $U_D(t)$.

\begin{theorem}\label{th-finite-diffusion}
Let $h$ be a random $n$-dimensional vector with zero expectation and variance $D$. Then
    $$
    \mathbb{E}_h \Av_m S_{\sqrt{t}h} 
    = \mathbb{E}_{h_1, \ldots, h_m} S_{ \sqrt{t / m}\, h_1} \circ \ldots \circ S_{\sqrt{t / m}\, h_m}
    $$
    at $m\to\infty$ converges uniformly over segments in a strong topology to the semigroup $\mathfrak{U}(t)$, which represents the solution of the equation
\begin{align*}
u_t' &= \dfrac{1}{2} \Delta_D u,\\
\left. u \right|_{t=0} &= u_0(x).
\end{align*}
    That is, for all $f\in H, t > 0$ is fulfilled
    $$
    \lim\limits_{m\to\infty} \sup\limits_{t\in [0,T]} \| \mathfrak{U}(t) f - \mathbb{E}_h \Av_m S_{\sqrt{t}h} f\|_H = 0.
$$
\end{theorem}

\begin{proof}
    Consider an operator-valued function
    $$
    V(t) = \mathbb{E}_h S_{\sqrt{t}h}.
    $$
    Then, by lemma \ref{expect-independent-unitary},
    $$
    \mathbb{E} \Av_m S_{\sqrt{t}\, h} 
    = \mathbb{E}_{h_1, \ldots, h_m} 
    (S_{\sqrt{t/m}\, h_1} \circ \ldots \circ S_{\sqrt{t/m}\, h_m})
    =(\mathbb{E}_{h} S_{\sqrt{t/m}\, h})^m 
    = (V(t/m))^m.
    $$
    When $t=0$, equality $S_{0 \cdot h} = I$ is fulfilled, hence $V(0) = I$. In addition, $\|S_{th}\| = 1$ due to unitarity, which means that the estimate of $\|V(t)\|\le 1$ is fair.

    Also, according to lemma \ref{lemma-finite-der-at-zero}, the derivative of $V(t)$ at zero is equal to $\frac{1}{2} \Delta_D$. Moreover, as we found out, this operator is densily defined, closable and generates a semigroup of convolution with a Gaussian measure with variance $\dfrac{1}{2}tD$.

    Using Chernov's theorem, we can prove the following result.
\end{proof}

\begin{theorem}\label{th-qosc-conseq-thifts}
    Let $h$ and $a$ be unbiased random vectors with variances $D_x$ and $D_p$, respectively. Then
    $$
    \mathbb{E}_{h, a} \Av_m (S_{\sqrt{t}h} \widehat{S}_{\sqrt{t}a})
    $$
    at $m\to\infty$ converges uniformly along the segments to the semigroup that represents the solution of the equation
    $$
    u'(t) = \dfrac{1}{2} \Delta_{D_x} u - \dfrac{1}{2} (D_p x, x) u.
    $$
\end{theorem}

\begin{proof}
Consider the following strongly continuous families of operators
\begin{align*}
V_1(t) &= \mathbb{E}_h S_{\sqrt{t}h},\\
V_2(t) &= \mathbb{E}_a \widehat{S}_{\sqrt{t}a},\\
V(t) &= \mathbb{E}_{h,a} S_{\sqrt{t}h} \widehat{S}_{\sqrt{t}a} = V_1(t) V_2(t).
\end{align*}
Obviously, $V(0) = I$ and $\|V(t)\|\le 1$.

By lemma \ref{expect-independent-unitary} we have
$$
\mathbb{E}_{h,a} \Av_m ( S_{\sqrt{t}h} \widehat{S}_{\sqrt{t}a}) = (V(t/m))^m.
$$

As shown in lemma \ref{lemma-finite-der-at-zero}, $V_2'(0)$ is a closable operator $-\dfrac{1}{2}(D_p x,x)$, the closure of which is the generator of the semigroup of multiplication by $\exp\left(-\dfrac{1}{2}(D_p x,x) \right)$. Applying the Fourier transform, we also obtain that the operator $\dfrac{1}{2}\Delta_{D_x}= V_2'(0)$ is closable, and its closure is a generator of the convolution semigroup with a Gaussian measure.
$$
V'(0) = V_1'(0) V_2(0) + V_1(0) V_2'(0) = \dfrac{1}{2}  \Delta_{D_x} - \dfrac{1}{2}  (D_p x,x).
$$

The space of test functions $\mathcal{D}(\mathbb{R}^n)$ can be chosen as a common essential domain for all operators. In order to use Chernov's theorem, we only need to show that the closure of the operator $\dfrac{1}{2}\Delta_{D_x} -\dfrac{1}{2} (D_p x,x)$ is a generator of some strongly continuous semigroup. To do this, we show that this operator is self-adjoint. Consider three sesquilinear forms
\begin{align*}
t_1(f,g) &= \left(\left(I - \dfrac{1}{2} \Delta_{D_x}\right) f, g\right), \quad D(t_1) = \mathcal{D}(\mathbb{R}^n),\\
t_2(f,g) &= \left(\left(I + \dfrac{1}{2} (D_p x, x)\right) f, g\right), \quad D(t_1) = \mathcal{D}(\mathbb{R}^n),\\
t_3(f,g) &= \left(\left(I - \dfrac{1}{2} \Delta_{D_x} + \dfrac{1}{2} (D_p x, x)\right) f, g\right), \quad D(t_2) = \mathcal{D}(\mathbb{R}^n).
\end{align*}
Since $t_1\le t_3$ and $t_2\le t_3$, by the theorem \ref{th-repr-operators-incl} we get that the Friedrichs extension of the operator $\dfrac{1}{2}\Delta_{D_x} - \dfrac{1}{2} (D_p x,x)$ is defined at the intersection of the closure domains of the operators $\Delta_{D_x}$ and $(D_p x,x)$. Therefore, it coincides with the closure of the operator $\dfrac{1}{2}\Delta_{D_x} - \dfrac{1}{2} (D_p x,x)$, which guarantees the self-conjugacy of the latter.

Let us take an arbitrary function $f$ from the domain of the Friedrichs extension. As we obtained, $\Delta_{D_x} f$ and $-(D_p x, x)f$ are correctly defined. Consider the function $f_t =U_{tD_x} \widehat{U}_{tD_p} f\in\mathcal{D}(\mathbb{R}^n)$. It is evident that for any $k$ one has
$$
\|x_k (f - f_t)\|_{L_2(\mathbb{R}^n)} \to 0, \quad t \to \infty,
$$
and
$$
\|\partial_k (f - f_t)\|_{L_2(\mathbb{R}^n)} \to 0, \quad t \to \infty,
$$
from that it follows that $\left(\dfrac{1}{2}  \Delta_{D_x} - \dfrac{1}{2}  (D_p x,x)\right) f_t(x)$ converges to $\left(\dfrac{1}{2}  \Delta_{D_x} - \dfrac{1}{2}  (D_p x,x)\right) f_t(x)$. 
\end{proof}

\begin{theorem}
Let $h$ and $a$ be unbiased random vectors with variances $D_x$ and $D_p$ accordingly, and $p\in (0,1)$. Let us introduce the operator $T_{p,h,a}$, which with probability $p$ equals to a random shift in the coordinate representation $S_{h}$, and with probability $1-p$ to a random shift in the momentum representation $\widehat{S}_{a}$. Then the sequence
$$
\mathbb{E}_{h, a} \Av_m (T_{p,\sqrt{t}h,\sqrt{t}a})
$$
at $m\to\infty$ converges uniformly along the segments to the semigroup that represents the solution of the equation
$$
u'(t) = \dfrac{p}{2} \Delta_{D_x} u - \dfrac{1-p}{2} (D_P x, x).
$$
\end{theorem}
\begin{proof}
The proof is similar to the proofs of the theorems \ref{th-finite-diffusion} and \ref{th-qosc-conseq-thifts}. Let us introduce an operator-valued function
$$
V(t) = \mathbb{E} T_{p,\sqrt{t} h,\sqrt{t} a}.
$$
Due to the linearity of the expectation, one has
\begin{align*}
V(t) &= p \cdot \mathbb{E} S_{\sqrt{t}h} + (1-p) \cdot \mathbb{E} \widehat{S}_{\sqrt{t} a},\\
V'(0) &= \dfrac{p}{2} \Delta_{D_x} u - \dfrac{1-p}{2} (D_P x, x).
\end{align*}
As before,
$$
\mathbb{E}_{h, a} \Av_m (T_{p,\sqrt{t}h,\sqrt{t}a}) = (V(t/m))^m
$$
and we can use the Chernov's theorem.
\end{proof}

\section{Square integrable functions on the Hilbert space and shifts in the coordinate and momentum representations}\label{Sec:Shifts}

In this section we define a finite additive translation invariant measure $\lambda_\mathcal{E}$ on a real separable Hilbert space $E$ with an orthonormal basis $\mathcal{E}$. On the space of square integrable functions with respect to $\lambda_\mathcal{E}$, we consider shifts in the coordinate and momentum representations. We study properties of the uniraty groups of shifts along vector $th, t\in \mathbb{R}, h\in E$: we find criteria for their strong continuity and show how they can be approximated by shifts along finite vectors.

Let $D$ be a non-negative nuclear operator in a real separable Hilbert space $E$, having a basis of eigenvectors ${\cal E}=\{e_k\}$ and the corresponding set of eigenvalues $\{d_k\}$.

\begin{definition}
    Recall (\cite{S2016, SITO}) that a block with edges, collinear to vectors of the basis $\cal E$ is a set of the form
    $$
    \Pi _{a,b}=\{ x\in E:\ (e_j,x)\in [a_j,b_j) \ \forall \ j\in {\mathbb N}\},
    $$
    where $a,b\in l_{\infty }$ and $a_j<b_j \ \forall \ j\in {\mathbb N}$. 
    
    A block $\Pi _{a,b}$ is called absolutely measurable if either $\Pi_{a,b}=\emptyset$, or the condition is hold
    $$
    \sum\limits_{j=1}^{\infty }\max \{ 0,\ln(b_j-a_j)\}<+\infty.
    $$
\end{definition}

On the set of all absolutely continuous blocks with edges, collinear to vectors of the basis $\cal E$, a non-negative function of the set $\lambda_{\cal E}$ is given, defined by equality
    $$
    \lambda _{\cal E}(\Pi _{a,b})=
    \begin{cases}
        \prod\limits_{j=1}^{\infty }(b_j-a_j), & \text{if } \Pi _{a,b}\neq \emptyset, \\ 
        0, & \text{otherwise}.
    \end{cases}
    $$
The function $\lambda_{\cal E}$ has a single extension $\lambda_{\cal E}$ to a finitely additive measure on the minimal ring $r_{\cal E},$ containing the set of all absolutely measurable blocks, which is uniquely determined to a measure on the complection of ${\cal R}_{\cal E}$  of the ring $r_{\cal E}$ over measure $\lambda_{\cal E}$.

Consider the set $S(E,{\cal R}_{\cal E},{\mathbb C})$ of linear combinations of indicator functions of sets from the ring ${\cal R}_{\cal E}$. Every such a function has the form
$$
f(x) = \sum\limits_{k=1}^m c_k\chi _{B_k},
$$
where $m\in\mathbb{N}, c_k\in {\mathbb C}$, and measurable sets $B_k\in {\cal R}_{\cal E}$ without loss of generality can be considered as pairwise disjoint. Equip $S(E,{\cal R}_{\cal E},{\mathbb C})$ with the sesquilinear form $p_2$, which, if $B_k$ are chosen to be pairwise disjoint, is equal to
\begin{equation}\label{nor}
    p_2 \left(
        \sum\limits_{k=1}^m c_k\chi _{B_k}
    \right)
    =\left(
        \sum\limits_{k=1}^m|c_k|^2\lambda _{\cal E}({B_k}) \right)^{1/2}.
\end{equation}

\begin{definition}
    Equipped with norm (\ref{nor}) quotient space
    $$
    S_2(E,{\cal R}_{\cal E},\lambda _{\cal E},{\mathbb C})=S(E,{\cal R}_{\cal E},{\mathbb C}) / {\cal N}_2,
    $$
    where 
    $$
    {\cal N}_2=\{ u\in S(E,{\cal R}_{\cal E},{\mathbb C}):\ p_2(u)=0\}
    $$
    is a pre-Hilbert space, whose completion is the Hilbert space of square integrable functions ${\cal H}_{\cal E}$.
\end{definition}

Let $E_N={\rm span}(e_1,\ldots,e_N)$ and $E^N=(E_N)^{\bot}$, and let ${\cal E}^{N}=\{ e_{N+1},\ldots\}$ be an orthonormal basis in the space $E^N$ of the eigenvectors of the operator $D$.

Let $\lambda _{E_N}$ be the Lebesgue measure on the space $E_N$, and $\lambda_{{\cal E}^N}$ be a shift-invariant measure on the space $E^N$, constructed using ONB ${\cal E}^N$ in the same way as the measure  $\lambda _{\cal E}$. Let us denote by $P_N$ the projection operator to the first $N$ coordinates.

According to theorem 4 of \cite{BusS}, the following equality holds
$$
{\cal H}_{\cal E}=L_2(E_N)\otimes {\cal H}_{{\cal E}^N}.
$$
Since the unitary Fourier transform $\mathcal{F}_{E_N}$ is defined on $L_2(E_N)$,  it has the unique extension to a unitary operator $\mathcal{F}_n =\mathcal{F}_{E_N}\times I$ on ${\cal H}_{\cal E}$.

Recall that by $B({\cal H}_{\cal E})$ we denote the Banach space of bounded linear operators on the Hilbert space ${\cal H}_{\cal E}$ endowed with the operator norm.

\begin{definition}
    For each vector $h\in E$, we define a unitary shift operator $S_h \in B({\cal H}_{\cal E})$ as
    $$
    S_h f(x) = f(x-h),\ f \in {\cal H}_{\cal E}.
    $$
\end{definition}

For a fixed vector $h$, the parametric family of shift operators $S_{th},\ t\in\mathbb{R}$ forms a unitary group.

Let us introduce the space $L_1({\mathcal E})$ as a subspace of $E$ consisting of vectors whose sequence of coordinates in the basis $\mathcal E$ belongs $l_1$. Let us endow this space with the norm
$$
\|h\|_{L_1(\mathcal{E})} = \sum\limits_{k=1}^\infty |h_k| = \sum\limits_{k=1}^\infty |(h, e_k)|.
$$
\begin{theorem}\label{lemma-shift-continuous}[section 4 of \cite{S2016}]
    The group $S_{th}$ is strongly continuous if and only if $h\in L_1({\mathcal E})$.
\end{theorem}
\begin{proof}
    Let $\forall k \in \mathbb{N}: 1 + \delta_k > 0$.

    Direct calculations show that 
    $P = \sum\limits_{k=1}^\infty\ln(1 + \delta_k)$
    absolutely converges to a positive number if and only if the series 
    $S = \sum\limits_{k=1}^\infty \delta_k$ 
    absolutely converges.

    As a result, a block with sides $[a_k,b_k)$ is measurable and has a positive measure if and only if the series $\sum\limits_{k=1}^\infty (b_k-a_k-1)$ converges absolutely.

    Since the operator $S_{th}$ is bounded for any $t$ and $h$, if it can be shown that $S_{th}$ is strongly continuous on some dense subspace ${\cal H}_{\cal E}$, then it is strong continuous on all ${\cal H}_{\cal E}$. Since the characteristic functions of measurable blocks generate such a dense subspace, it is enough to check strong continuity only for them. On the other hand, if for some measurable block $Q$ the function $S_{th} \chi_Q$ is not continuous in $t$, then this means that there is no strong continuity of $S_{th}$.

    Consider an arbitrary measurable block $Q$ with sides $[a_k,b_k)$.
    $$
    \|S_{th} \chi_Q - \chi_Q\|_{{\mathcal H}_\mathcal{E}}^2
    = \|S_{th} \chi_Q \|_{{\mathcal H}_\mathcal{E}}^2 + \|\chi_Q\|_{{\mathcal H}_\mathcal{E}}^2 
        - 2 (S_{th} \chi_Q , \chi_Q )_{{\mathcal H}_\mathcal{E}}
    = 2 \lambda_\mathcal{E}(Q) - 2 \lambda_\mathcal{E}((Q + th) \cap Q)
    $$
    Since the block is measurable, then $b_k - a_k > 0$ and $\lim\limits_{k\to \infty} b_k - a_k = 1$, which means $\min\limits_k (b_k - a_k)$ is positive and is achieved at some $K\in\mathbb{N}$. By limiting $t$ to a sufficiently small neighborhood of zero, we can guarantee that all intersections of $[a_k + t h_k, b_k + t h_k] \cap [a_k, b_k]$ are non-empty. In this case, this intersection is a half-interval of $[a_k + t |h_k|, b_k)$ with a positive $h_k$ and a half-interval $[a_k, b_k - t |h_k|)$ if negative. Anyway,
    $$
    \lambda_\mathcal{E}((Q + th) \cap Q) = \prod\limits_{k=1}^\infty (b_k - a_k - t |h_k|).
    $$
    Since the block $Q$ has a positive measure, $\prod\limits_{k=1}^\infty (a_k - b_k)$ must take a positive value, which means that the series $\sum\limits_{k=1}^\infty b_k - a_k - 1$ converges absolutely. If $h\notin L_1(\mathcal{E})$, then $\sum\limits_{k=1}^\infty b_k - a_k - t|h_k| - 1$ diverges for any $t$, and the series $\sum\limits_{k=1}^\infty \ln(b_k - a_k - t|h_k|)$ does not converge absolutely, which means that the block $(Q + th) \cap Q$ has zero measure, $\|S_{th} \chi_Q - \chi_Q\|_{{\mathcal H}_\mathcal{E}}^2 =2\lambda_\mathcal{E}(Q)$, and $S_{th}$ is discontinuous in a strong topology.

    Now assume $h\in L_1(\mathcal{E})$. Then, due to the unitarity of the shifts, the following estimate is valid (recall that the operator $P_k:E\to E$ denotes the projection onto the first $k$ coordinates, and $P_0$ is set to zero):
\begin{align*}
    \|S_{th} f(x) - f(x)\|^2_{\mathcal{H}_\mathcal{E}}
    &\le \sum\limits_{k=1}^\infty 
    \|S_{P_{k} th} \chi_Q(x) - S_{P_{k-1} th} \chi_Q(x)\|^2_{\mathcal{H}_\mathcal{E}}
    = \sum\limits_{k=1}^\infty
    \|S_{th_k} \chi_Q(x) - \chi_Q(x)\|^2_{\mathcal{H}_\mathcal{E}} \\
 &= \sum\limits_{k=1}^\infty \left[
        2 \lambda_\mathcal{E}(Q) - 2 \lambda_\mathcal{E}((Q + th_k) \cap Q) 
    \right]
    \le 2 \lambda_\mathcal{E}(Q)
    \sum\limits_{k=1}^\infty 
        \dfrac{t|h_k|}{b_k - a_k}.
\end{align*}
    Since $\min\limits_k \{b_k - a_k\} > 0$, the expression on the right side is finite for any $t$ and tends to zero for $t\to 0$.
\end{proof}

\begin{theorem}\label{theorem-shift-appr}
    The sequence of operators $S_{P_n h}$ converges to $S_h$ in a strong operator topology if and only if $h\in L_1(\mathcal{E})$. Moreover, if $h\in L_1(\mathcal{E})$, then for all $T>0$ and $f\in\mathcal{H}_\mathcal{E}$

    $$\lim\limits_{m\to\infty} \sup\limits_{t \in [0,T]} \|S_{P_m th} f(x) - S_{th}f(x)\| = 0.$$
\end{theorem}
\begin{proof}
    Suppose that $h\in E$. 
    As in the proof of the previous lemma, we take an arbitrary block $Q$ of positive measure with sides $[a_k, b_k)$.
    $$
    \|S_h \chi_Q - S_{P_m h} \chi_Q\|^2 _{\mathcal{H}_\mathcal{E}}
    = \|(S_h - S_{P_m h})  \chi_Q -  \chi_Q\|^2 _{\mathcal{H}_\mathcal{E}}
    = 2 \lambda_\mathcal{E}(Q) - 2 \lambda_\mathcal{E}((Q + (h - P_m h)) \cap Q).
    $$
    Let us choose a sufficiently large $m$ so that the intersection is non-empty, and write down its measure
    $$
    \lambda_\mathcal{E}((Q + (h - P_m h)) \cap Q) = \prod\limits_{k=1}^m (a_k - b_k)\cdot \prod\limits_{k=m+1}^\infty (a_k - b_k - |h_k|).
    $$
    Due to the measurability of the block $Q$, the series $\sum\limits_{k=m+1}^\infty b_k - a_k - 1$ converges absolutely.

    If $h\notin L_1(\mathcal{E})$, then the series $\sum\limits_{k=m+1}^\infty (b_k - a_k -|h_k| - 1)$ tends to $-\infty$, which means that the intersection measure $\lambda_\mathcal{E}((Q + (h - P_m h)) \cap Q)$ tends to zero. Then $\lambda_\mathcal{E}((Q+(h - P_m h)) \cap Q) = 0 $, and $\|S_h \chi_Q - S_{P_m h} \chi_Q\|_{\mathcal{H}_\mathcal{E}} = 2\lambda_\mathcal{E}(Q)$ does not tend to zero, which means that $S_{P_n h} f(x)$ does not converge to $S_{h}f(x)$.

    If $h\in L_1(\mathcal{E})$, then, as shown in the proof of the previous theorem,
    $$
    \|S_{th} \chi_Q - \chi_Q\|^2 _{\mathcal{H}_\mathcal{E}} 
    \le
    2 \lambda_\mathcal{E}(Q)
    \sum\limits_{k=1}^\infty 
        \dfrac{t|h_k|}{b_k - a_k}.
    $$
    Therefore,
    $$\sup\limits_{t \in [0,T]} \|S_{P_m th} \chi_Q - S_{th} \chi_Q\|^2 _{\mathcal{H}_\mathcal{E}}
    \le 
    2 \lambda_\mathcal{E}(Q)
    \sum\limits_{k=m+1}^\infty 
        \dfrac{T|h_k|}{b_k - a_k} \to 0, \quad m\to \infty.
    $$
\end{proof}

If all the coordinates of the vector $h\in L_1(\mathcal{E})$ starting from the index $m+1$ are equal to zero, then the Fourier transform $\mathcal{F}_m$ along the first $m$ coordinate directions sets its unitary equivalence with the shift operator in the momentum representation along the first $m$ coordinates:
$$
\widehat{S}_{th} f(x) = e^{it(h,x)} f(x).
$$
For any natural $n < m$, the unitary transformation $\mathcal{F}_m$ maps the difference $\widehat{S}_{P_m h} - \widehat{S}_{P_n h}$ into $S_{P_m h} - S_{P_n h}$. Therefore the sequence $\widehat{S}_{P_n h}$ is fundamental for $h \in L_1(\mathcal{E})$. This allows us to determine the shift in the momentum representation by an arbitary vector from $L_1(\mathcal{E})$.

\begin{definition}
    For any $h\in L_1(\mathcal{E})$, we define the operator
    $$
    \widehat{S}_{h} f(x) = e^{i(h,x)} f(x)
    $$
    as the limit of the sequence $\widehat{S}_{P_n h} f(x)$ at $n\to\infty$.
\end{definition}

\begin{theorem}
For any $h\in L_1(\mathcal{E})$ the semigroup $\widehat{S}_{th}$ is unitary and strongly continuous. In addition, for any $T>0$ and $f\in\mathcal{H}_\mathcal{E}$ one has
$$
\lim\limits_{m\to\infty} \sup\limits_{t \in [0,T]} \| \widehat{S}_{th} f(x) - \widehat{S}_{t P_m h} f(x)\| _{\mathcal{H}_\mathcal{E}} = 0.
$$
\end{theorem}
\begin{proof}
The unitarity is checked directly: operator $\widehat{S}_{h}^{-1} =\widehat{S}_{-h}$ is obtained from $\widehat{S}_{h}$ by complex conjugation.

The strong continuity of $\widehat{S}_{th}$ follows directly from the segment-uniform convergence of the semigroups $\widehat{S}_{t P_m} h$ in a strong operator topology, since any semigroup $\widehat{S}_{t P_m h}$ is strongly continuous.

Let us prove the latter property. As in the case of the theorem \ref{theorem-shift-appr}, it is enough to prove it only for the characteristic functions of blocks of a positive measure. Consider such a block $Q$ with sides $[a_k, b_k)$.
\begin{align*}
    \sup\limits_{t \in [0,T]} 
     \|\widehat{S}_{th}(t) f - \widehat{S}_{P_m th}(t) f\|
     _{\mathcal{H}_\mathcal{E}}
    &= \sup\limits_{t \in [0,T]} 
    \| f(x) - e^{it(x,h-P_m h)} f(x)\| _{\mathcal{H}_\mathcal{E}}\\
    &\le 
    \|f(x)\| 
    \sup\limits_{x \in Q, t \in [0,T]} |1 - e^{it(x,h - P_m h)}|.
\end{align*}
Since the block $Q$ is measurable, the sequence of its sides  $[a_k,b_k)$ has the property that $\lim\limits_{k\to\infty} (b_k - a_k) = 1$. In addition, in order for $Q$ to be non-empty, it is necessary that the sequence $\inf\limits_{x\in [a_k,b_k)} [a_k, b_k)$ converges to zero. Therefore, there exists such a $N_1 \ge N$ that for all coordinates $k \ge N_1$, $\max\limits_{x\in Q} |x_k| \le 2$. In this case, for a sufficiently large $m$
$$
\sup\limits_{x \in Q, t \in [0,T]} 
|1 - e^{it(x,h - P_m h)}| 
\le 
\left|1 - \exp\left(2iT \sum\limits_{k=m+1}^\infty h_k \right) \right|.
$$

Since $h\in L_1(\mathcal{E})$, this difference tends to zero at $m\to\infty$.
\end{proof}

\section{Averaging of random shifts}
\label{Sec:Averaging}

In this section we introduce semigroups of averaging operators of shift $S_{th}$ and $\widehat{S}_{th}$ over random $h$. We find the conditions when these semigroups are strongly continuous and when they can be approximated with averaging over shifts along finite vectors.

Let $D$ be a nuclear positive-definite self-adjoint operator whose eigenvectors coincide with the basis vectors ${\mathcal E}$. Consider on $E$ a one-parameter family of Gaussian measures $\nu_{tD}, \ t\ge 0$, such that for every non-negative $t$ the Fourier transform of the measure $\nu_{tD}$ is equal to
$$
\widehat{\nu}_t(\xi) = e^{-\frac{1}{2} t D \xi^2}.
$$
In other words, consider a family of Gaussian measures with zero expectation and variance $tD$ (according to Theorem 2.1\cite{Kuo75}, the nuclearity of the operator $D$ is equivalent to the finiteness of the second moment of the measure $\nu_{tD}$). The constructed one-parameter family of Gaussian measures forms a semigroup with respect to the convolution operation~\cite{BogGM},

\begin{equation}\label{gauss-semigroup}
\nu_{tD} \ast \nu_{sD} = \nu_{(t+s)D}.
\end{equation}

Recall~\cite{BogGM,path_integrals} that the Gaussian measure $\nu$ on the space $E$ is a finite additive function of a set on the minimal $\sigma$-algebra $C(E)$ containing cylindrical subsets of $E$ and such that its projection  on any finite-dimensional subspace of $E$ is a Gaussian measure. In this case, according to theorems 2.1 and 2.3 in~\cite{Kuo75},  the set of finitely additive functions has a unique countably additive extension to a measure on the $\sigma$-algebra ${\mathcal B}(E)$ of Borel subsets of $E$ if and only if it has a nuclear covariance operator $D$.

Using defined one-parameter family of Gaussian measures $\nu_{tD}$, we define one-parameter families of transformations of the space ${\mathcal H}_\mathcal{E}$ as a result of averaging the transformation of random shift operators $S_h$ and $\widehat{S}_h$ with respect to $h\in E$ distributed as $\nu_{tD}$.

\begin{definition}\label{def-smoothed-shift}
    Let $D$ be a nuclear positive-definite self-adjoint operator whose eigenvectors coincide with the basis vectors ${\mathcal E}$. We  define
    $$
    {\mathcal U}_D(t) f(x) 
    = \int\limits_E S_h f(x) d \nu_{tD}(h) 
    $$
    and, in the case when $D^{1/2}$ is also nuclear,
    $$
    \widehat{{\mathcal U}}_D(t) f(x) 
    = \int\limits_E \widehat{S}_h f(x) d \nu_{tD}(h),
    $$
    where the integral is understood in the sense of Pettis.
\end{definition}

The requirement of the nuclearity of $D^{1/2}$ is motivated by the fact that the shift in the momentum representation is defined only for vectors from $L_1(\mathcal{E})$, and the nuclearity of $D^{1/2}$ is equivalent to the property $\nu_D(L_1(\mathcal{E})) = 1$ (\cite{sonis}). However, as we see later, $\widehat{{\mathcal U}}_D(t)f(x)$ can also be defined in the case when $D^{1/2}$ is not nuclear, as the limit of averaging in finite-dimensional subspaces.

By virtue of the equality~(\ref{gauss-semigroup}), both defined families are semigroups. Also, for any $t$, the operators ${\mathcal U}_D(t)$ and $\widehat{{\mathcal U}}_D(t)$ are self-adjoint.

For fixed $f$ and $g$, the numeric functions $\phi(h) = (S_h f(x), g(x))$ and $\psi(h) = (\widehat{S}_h f(x), g(x))$ are continuous and bounded, and therefore measurable with respect to the $\sigma$-algebra $C(E)$ generated by cylindrical sets, which guarantees the correctness of the definition. At the same time, both of these functions do not exceed by the norm $\|f(x)\|_{\mathcal{H}_\mathcal{E}} \|v(x)\|_{\mathcal{H}_\mathcal{E}}$, therefore ${\mathcal U}_D(t)$ and $\widehat{{\mathcal U}}_D(t)$ are bounded linear operators $\mathcal{H}_\mathcal{E}\to\mathcal{H}_\mathcal{E}$ with norm less or equal to one.

Note also that by virtue of its definition, the operator $\mathcal{U}_D(t)= \mathcal{U}_{tD}$ commutes with any shift operator $S_r$, since
\begin{align*}
({ \mathcal U}_D(t) S_r u(x), v(x))
    &= \int\limits_E ( S_h S_r u(x), v(x)) d \nu_{tD}(h)
    = \int\limits_E (S_h u(x), S_{-r} v(x)) d \nu_{tD}(h) \\
    &= ({\mathcal U}_D(t) u(x), S_{-r} v(x))
    = ({S_r \mathcal U}_D(t) u(x), v(x)).
\end{align*}
Similarly, $\widehat{\mathcal{U}}_{tD}$ commutes with an arbitrary operator $\widehat{S}_r$. 

Consider the operator $D_N = D P_N$, in which the first $N$ eigenvalues coincide with $D$, and the rest are zero. The restriction of $D_N$ to $E_N$ is the matrix $n\times n$, which we denote by $D_{E_N}$. Then both operators ${\mathcal U}_{tD_N}$ and $\widehat{{\mathcal U}}_{tD_N}$ act on $E=E_N\times E^N$ as $U_{t D_{E_N}}\times I$ and $\widehat{U}_{t D_{E_N}} \times I$ respectively.
Here $U_{t D_{E_N}}$ denotes convolution with $n$-dimensional Gaussian measure with zero expectation and variance $D_{E_N}$, and $\widehat{U}_{t D_{E_N}}$ represents multiplication by $\exp(-t(D_{E_N} x, x))$.

So we get that
\begin{align*}
    {\mathcal U}_{tD_N} f(x) 
    &=  f(x) \star (\nu_{tD_N} \otimes \delta(0)),\\
    \widehat{{\mathcal U}}_{tD_N} f(x) 
    &= \exp(-(D_N x, x)) f(x).
\end{align*}
    
Let us first consider the properties of the coordinate-wise smoothing semigroup.

\begin{theorem}\label{th-coord-smooth-appr}
Let $D$ be a nuclear positive-definite self-adjoint operator whose eigenvectors coincide with the basis vectors ${\mathcal E}$. If $D^{1/2}$ is nuclear, then for any $T>0$ and $f\in\mathcal{H}_\mathcal{E}$ one has
$$
\lim\limits_{m\to\infty}
\sup\limits_{t \in [0,T]}
\| {\mathcal U}_{tD} f(x) - {\mathcal U}_{tD_m} f(x)\| _{\mathcal{H}_\mathcal{E}} = 0.
$$
\end{theorem}
\begin{proof}
Since norms of all operators ${\mathcal U}_{tD}$ and ${\mathcal U}_{tD_N}$ do not exceed one, it is enough to show the required convergence for all characteristic functions $f= \chi_Q$ of beams of finite measure, since they form a dense subset in $\mathcal{H}_\mathcal{E}$.

We use the fact that $D^{1/2}$ is nuclear if and only if $\nu_D(L_1(\mathcal{E})) = 1$ (\cite{sonis}).

Fix $\varepsilon > 0$ and consider in $(L_1(\mathcal{E}))$ a family of subsets
$$
X_m = 
\{ h \in L_1(\mathcal{E}) \mid \|h - P_m h\|_{L_1(\mathcal{E})} < \varepsilon
    \}.
    $$
    Because for all
    $h \in L_1( \mathcal{E})$
    the value
    $\sum\limits_{k=m+1}^\infty |h_k|$
    tends to zero at
    $m\to\infty$,
    vector $h$ lies in all $X_m$ starting from some $m$.
    Thus,
    $L_1(\mathcal{E}) = \bigcup\limits_{m=1}^\infty X_m$,
    and $\lim\limits_{m\to\infty} \nu_D(X_m) = 1$.

    For any measurable block $Q$ with sides $[a_k,b_k)$,
    $$
    \|S_{h} \chi_Q - \chi_Q\|^2 \le 2 \lambda_\mathcal{E}(Q)
    \sum\limits_{k=1}^\infty 
        \dfrac{|h_k|}{b_k - a_k} = C(Q) \|h\|_{L_1(\mathcal{E})}.
$$
Thus,
$$
    \left| \int\limits_{X_m} |S_{th} \chi_Q - S_{P_m th} \chi_Q|^2  d \nu_D(h) \right|
    \le 
    \int\limits_{X_m} 
    \|S_{t(h - P_m h)} \chi_Q - \chi_Q\|^2_{\mathcal{H}_\mathcal{E}} d \nu_D(h)
    \le 
$$
$$
    \le
    \int\limits_{X_m} C(Q) t \|h - P_m h\|_{L_1(\mathcal{E})} d \mu_D(h)
    \le C(Q) t \varepsilon \nu_D(X_m)
    \le C(Q) t \varepsilon.
$$
On the other hand,
    $$
    \left|\ \int\limits_{L_1(\mathcal{E}) \setminus X_m} |S_{th} \chi_Q - S_{P_m th} \chi_Q|^2  d \nu_D(h) \right|
    \le 2 \lambda_\mathcal{E}(Q) \nu_D(L_1(\mathcal{E}) \setminus X_m) \to 0, \quad m \to \infty.
    $$
    Thus, we have obtained that for any given $\varepsilon > 0$ the difference
    $\|{\mathcal U}_{tD}\chi_Q - {\mathcal U}_{tD_N} \chi_Q\|_{\mathcal{H}_\mathcal{E}}$ does not exceed $C(Q) t\varepsilon+ o(m)$, which implies the necessary convergence.
\end{proof}

\begin{corollary}\label{cor-coord-smooth-cont}
    Let $D$ be a nuclear positive-definite self-adjoint operator whose eigenvectors coincide with the basis vectors ${\mathcal E}$. If $D^{1/2}$ is nuclear, then ${\mathcal U}_{tD}$ is strongly continuous.
\end{corollary}

    Note that the requirement for the nuclearity of $D^{1/2}$ in the theorem \ref{th-coord-smooth-appr} and the corollary \ref{cor-coord-smooth-cont} can not be relaxed.

\begin{theorem}\label{th-coord-smooth-trivial}
    Let $D$ be a nuclear positive-definite self-adjoint operator whose eigenvectors coincide with the basis vectors ${\mathcal E}$. If $D^{1/2}$ not be nuclear, then for any $T>0$ and $f\in\mathcal{H}_\mathcal{E}$ one has $\mathcal{U}_{tD} f = 0$. At the same time, for any $f\in\mathcal{H}_\mathcal{E}$, the sequence of functions $\mathcal{U}_{D_n} f$ converges to zero, that is, to $\mathcal{U}_D(t) f$.
\end{theorem}
\begin{proof}
    Consider two measurable blocks of positive measure $Q$ and $W$ with sides $[a_k, b_k)$ and $[r_k, s_k)$, respectively, and consider the scalar product
    $$
    (\mathcal{U}_D \chi_Q, \chi_W) _{\mathcal{H}_\mathcal{E}}
    = \int\limits_E (S_h \chi_Q, \chi_W) _{\mathcal{H}_\mathcal{E}} d \nu_D(h)
    = \prod_{k=1}^\infty \int\limits_{r_k}^{s_k} \chi_{[a_k, b_k]}(x) \star \rho_{d_j}(x) dx,
    $$
    where $\rho_{d_j}(x)$ is the density function of a normal random variable with zero expectation and variance $d_j$.

    Each of these integrals reaches the largest value with a fixed difference $s_k-r_k$ if the center of the half-interval is $[r_k, s_k)$ coincides with the center of the half-interval $[a_k, b_k)$. In this case, we denote $[r_k, s_k) = [a_k - \delta_k, b_k + \delta_k)$. It follows from the measurability of $W$ that the series of $\delta_k$ absolutely converges. We consider all $\delta_k$ as positive, because in this case $(\mathcal{U}_D\chi_Q, \chi_W)$ will be the maximum.
\begin{align*}
    (\mathcal{U}_D \chi_Q, \chi_W) _{\mathcal{H}_\mathcal{E}} 
    &= \int\limits_E (S_h \chi_Q, \chi_W) _{\mathcal{H}_\mathcal{E}} d \nu_D(h) \\
    &= \prod_{k=1}^\infty \int\limits_{\mathbb{R}} \mu_J([b_k + x, a_k + x) \cap [b_k + \delta_k, a_k - \delta_k)) \rho_{d_j}(x) dx\\
    & \le \prod_{k=1}^\infty \int\limits_{\mathbb{R}} \max(0, b_k - a_k + 2\delta_k - |x|) \rho_{d_j}(x) dx \\
    & = \prod_{k=1}^\infty 2 \cdot \int\limits_0^{b_k - a_k + 2 \delta_k}  (b_k - a_k + 2\delta_k - x) \rho_{d_j}(x) dx.
 \end{align*}
    Substitute $x=\sqrt{d_k} y$ in each integral to switch to the standard Gaussian value.
    $$
    (\mathcal{U}_D \chi_Q, \chi_W) _{\mathcal{H}_\mathcal{E}}
    = \prod_{k=1}^\infty 2 \cdot \int\limits_0^{(b_k - a_k + 2 \delta_k)/\sqrt{d_k}}  (b_k - a_k + 2\delta_k - \sqrt{d_k}y) \rho(y) dy.
    $$
    Omitting the technical calculations, we get
     $$
     (\mathcal{U}_D \chi_Q, \chi_W) = \prod\limits_{i=1}^\infty \left[ 
    (b_i - a_i + 2\delta_i) {\rm erf} \left( \dfrac{b_i - a_i + 2\delta_i}{ \sqrt{2 d_i}} \right)
    + 2 \sqrt{d_i}  \rho \left( \dfrac{b_i - a_i + 2\delta_i }{\sqrt{d_i}} \right) - \sqrt{d_i} \sqrt{\dfrac{2}{\pi}}
    \right].
    $$
    When $i \to \infty$ value ${\rm erf} \left( \dfrac{b_i - a_i + 2\delta_i}{ \sqrt{2 d_i}} \right) - 1$ and $\rho \left( \dfrac{b_i - a_i + 2\delta_i }{\sqrt{d_i}} \right)$ converges to zero faster than any polynomial over $d_i$, and asymptotically 
\begin{align*}
\alpha_i &= (b_i - a_i + 2\delta_i) {\rm erf} \left( \dfrac{b_i - a_i + 2\delta_i}{ \sqrt{2 d_i}} \right)
    + 2 \sqrt{d_i}  \rho \left( \dfrac{b_i - a_i + 2\delta_i }{\sqrt{d_i}} \right) -  \sqrt{d_i} \sqrt{\dfrac{2}{\pi}}\\
&= b_i - a_i + 2 \delta_i - \sqrt{d_i}  \sqrt{\dfrac{2}{\pi}} + o(d_i^\infty),\quad i \to \infty.
\end{align*}
    It follows from the measurability of $Q$ and $W$ that $(b_i - a_i - 1)$ and $2\delta_i$ lie in $L_1(\mathcal{E})$, while $\{\sqrt{d_i}\}$  does not lie in $L_1(\mathcal{E})$ according to the assumption of the theorem. Therefore, the sum of the series of $\alpha_i$ diverges to $-\infty$, and $(\mathcal{U}_D\chi_Q, \chi_W) _{\mathcal{H}_\mathcal{E}}$ is zero for any $Q$ and $W$.

    As a consequence, $(\mathcal{U}_D f, g) _{\mathcal{H}_\mathcal{E}} = 0$ for any simple $f$ and $g$, and since the latter are dense in $\mathcal{H}_\mathcal{E}$, then $(\mathcal{U}_D f, g) _{\mathcal{H}_\mathcal{E}} = 0$ for any square integrable functions. Therefore, $\mathcal{U}_D f = 0$ as an element of $\mathcal{H}_\mathcal{E}$.

    Based on the same considerations and considering that the norm of the averaging operator does not exceed one, in order to prove that the sequence $\mathcal{U}_{D_n} f$ converges to zero, it is enough to do this for $f=\chi_Q$, where $Q$ is a measurable bar of positive measure.
    $$
    \|\mathcal{U}_{D_n} f\|^2 _{\mathcal{H}_\mathcal{E}} = (\mathcal{U}_{D_n} f, \mathcal{U}_{D_n} f) _{\mathcal{H}_\mathcal{E}} = (\mathcal{U}_{2D_n} f, f) _{\mathcal{H}_\mathcal{E}}.
    $$
    Let us use the calculations performed above:
    $$
    (\mathcal{U}_{2D_n} f, f) _{\mathcal{H}_\mathcal{E}}
    = \prod\limits_{k=1}^n \left(b_i - a_i -  \sqrt{\dfrac{4 d_i}{\pi}} + o(d_i^\infty) \right)\cdot \prod\limits_{k=n+1}^\infty  (b_i - a_i).
    $$
    Since the sequence $\{\sqrt{d_k}\}$ diverges, then at $n\to\infty$ this expression tends to zero, which finishes the proof.
\end{proof}

Let us now consider momentum averaging. Namely, we show that the limit of finite-dimensional approximations $\widehat{\mathcal{U}}_{D_n}$ is correctly defined and is continuous for nuclear $D$, and if $D^{1/2}$ is also nuclear, then this limit is the averaging of shifts in the momentum representation of vectors. From this we can immediately conclude that the semigroups $\mathcal{U}_{D}$ and $\widehat{\mathcal{U}}_{D}$ are not unitarily equivalent, since for non-nuclear $D^{1/2}$ the first one is absolutely discontinuous in a strong operator topology, and the second one is continuous.

\begin{theorem}\label{th-mom-smooth-appr}
    Let $D$ be a nuclear positive-definite self-adjoint operator whose eigenvectors coincide with the basis vectors ${\mathcal E}$. Then the sequence $\widehat{\mathcal{U}}_{D_n}$ converges in the strong operator topology of the space to some operator $A_{D}$. Moreover, for any $T>0$ and $f\in\mathcal{H}_\mathcal{E}$ one has
    $$
    \lim\limits_{m\to\infty}
    \sup\limits_{t \in [0,T]}
    \|A_{tD} f(x) - \widehat{\mathcal{U}}_{tD_m} f(x)\| _{\mathcal{H}_\mathcal{E}} = 0.
    $$
\end{theorem}
\begin{proof}
    Again, it is enough to prove the theorem for arbitrary $f$ from some dense subspace of square integrable functions.
    
    Fix $T>0$, $\varepsilon > 0$ and some $f =\chi_Q$ for a measurable bar $Q$ of positive measure. For all $n<m$
\begin{align*}
    \sup\limits_{t \in [0,T]} \|\widehat{\mathcal{U}}_{tD_n} f - \widehat{\mathcal{U}}_{tD_m} f\| _{\mathcal{H}_\mathcal{E}}
    &= \sup\limits_{t \in [0,T]} 
    \left\| \exp\left( -\dfrac{t}{2} \sum\limits_{k=1}^m d_k x_k^2\right) f(x) - \exp\left( -\dfrac{t}{2} \sum\limits_{k=1}^n d_k x_k^2\right) f(x)  \right\| _{\mathcal{H}_\mathcal{E}}\\
    &\le \|f\| _{\mathcal{H}_\mathcal{E}} \cdot \max\limits_{x\in Q}\left| 1 - \exp\left( -\dfrac{T}{2} \sum\limits_{k=n+1}^m d_k x_k^2\right)\right|.
\end{align*}
    Since the block $Q$ is measurable, the sequence of its sides $[a_k,b_k]$ has the property that $\lim\limits_{k\to\infty} (b_k - a_k) = 1$. In addition, in order for $Q$ to be nonempty, it is necessary that the sequence $\min\limits_{x\in [a_k,b_k]} [a_k, b_k]$ converges to zero. Therefore, there exists such a $N_1$ that for all coordinates $k\ge N_1$, $\max\limits_{x\in Q} |x_k|\le 2$.

    Further, since the operator $D$ is nuclear, there exists $N_2$ such that for all $N_2 \le i <j$, one has $\sum\limits_{k=i+1}^j d_k\le\varepsilon$.

    Therefore, if $n\ge\max(N_1, N_2)$, then
    $$
    \max\limits_{x\in Q}\left|1 - \exp\left( -\dfrac{T}{2} \sum\limits_{k=n+1}^m d_k x_k^2\right)\right|
    \le \left|1 - \exp\left( -2 T \sum\limits_{k=n+1}^m d_k \right)\right|
    \le \left|1 - \exp\left( -2 T \varepsilon \right)\right|.
    $$
    This implies that the sequence $\widehat{\mathcal{U}}_{tD_n}$ is fundamental in the semi-norm
    $\phi_{T,f} = \sup\limits_{t\in [0,t]}\|f\|_{\mathcal{H}_\mathcal{E}}$. 
    If we denote the limit $\widehat{\mathcal{U}}_{tD_n}$ at $n\to\infty$ for $A_{tD}$, then $\widehat{\mathcal{U}}_{tD_n}$ converge to $A_{tD}$ with respect to $\phi_{T,f}$.
\end{proof}

\begin{theorem}\label{th-two-mom-smooth}
    Let $D$ be a nuclear positive-definite self-adjoint operator whose eigenvectors coincide with the basis vectors ${\mathcal E}$. If $D^{1/2}$ is nuclear, then the limit 
    $A_{D} = \lim\limits_{n\to\infty} \widehat{\mathcal{U}}_{D_n}$
    exists and coincides with $\widehat{\mathcal{U}}_{D}$.
\end{theorem}
\begin{proof}
    The proof is similar to the proof of the theorem~\ref{th-coord-smooth-appr}. It is enough to prove that for any measurable block $Q$ one has 
    $$ \lim\limits_{m\to\infty} 
     \widehat{\mathcal{U}}_{D_m} \chi_Q 
     = \widehat{\mathcal{U}}_{D} \chi_Q.
     $$
    Define the sets $X_m$:
    $$
    X_m = 
    \{ h \in L_1(\mathcal{E}) \mid \|h - P_m h\|_{L_1(\mathcal{E})} < \varepsilon
    \}.
    $$
    They are nested in each other by construction, and moreover, $L_1(\mathcal{E}) = \bigcup\limits_{m=1}^\infty X_m$.

    For any measurable bar $Q$ with sides $[a_k,b_k)$ the following inequality holds
    $$
    \|\widehat{S}_{h} \chi_Q - \chi_Q\|^2
     = \|S_{h} \chi_Q - \chi_Q\|^2
     \le  C(Q) \|h\|_{L_1(\mathcal{E})}.
     $$
     Thus, we can use the inequalities from the proof of theorem~\ref{th-coord-smooth-appr}:
     $$
    \left|\ \int\limits_{X_m} |\widehat{S}_{th} \chi_Q - \widehat{S}_{P_m th} \chi_Q|^2  d \nu_D(h) \right|
    \le 
    \int\limits_{X_m} C(Q) t \|h - P_m h\|_{L_1(\mathcal{E})} d \mu_D(h)
    \le C(Q) t \varepsilon \nu_D(X_m)
    \le C(Q) t \varepsilon,
    $$

    $$\left|\ \int\limits_{L_1(\mathcal{E}) \setminus X_m} |\widehat{S}_{th} \chi_Q - S_{P_m th} \chi_Q|^2  d \nu_D(h) \right|
    \le 2 \lambda_\mathcal{E}(Q) \nu_D(L_1(\mathcal{E}) \setminus X_m) \to 0, \quad m \to \infty.
    $$
    Thus, we have obtained that for any given $\varepsilon > 0$ the difference is $\|{\widehat{\mathcal U}}_{tD}\chi_Q - {\widehat{\mathcal U}}_{tD_m}\chi_Q\|_{\mathcal{H}_\mathcal{E}}$ does not exceed $C(Q) t \varepsilon +o(m)$.
\end{proof}

Thus, if $D^{1/2}$ is nuclear, then
    $$
    \lim\limits_{m\to\infty}
    \sup\limits_{t \in [0,T]}
    \|\widehat{\mathcal{U}}_{tD} f(x) - \widehat{\mathcal{U}}_{tD_m} f(x)\| _{\mathcal{H}_\mathcal{E}} = 0.
    $$
Using the result of the theorem~\ref{th-mom-smooth-appr} , we can extend the definition~\ref{def-smoothed-shift}.

\begin{definition}\label{def-smoothed-shift-2}
    Let $D$ be a nuclear positive-definite self-adjoint operator whose eigenvectors coincide with the basis vectors ${\mathcal E}$. If the operator $D^{1/2}$ is nuclear, then we define the operator $\widehat{\mathcal{U}}_{D}$ as the limit of the operators $\lim\limits_{n\to\infty}\widehat{\mathcal{U}}_{D_n}$ in the strong operator topology. The theorem \ref{th-two-mom-smooth} guarantees consistency of this definition with definition~\ref{def-smoothed-shift}.
\end{definition}

The analogue of \ref{th-mom-smooth-appr} immediately follows:

\begin{corollary}\label{cor-mom-smooth-cont}
    Let $D$ be a nuclear positive-definite self-adjoint operator whose eigenvectors coincide with the basis vectors ${\mathcal E}$. If the operator $D^{1/2}$ is nuclear, then the semigroup $\widehat{\mathcal{U}}_{tD}$ is strongly continuous.
\end{corollary}

\section{Semigroup generators}
\label{Sec:Generators}

In this section, we study properties of generators of semigroups $\mathcal{U}_{tD}$ and $\widehat{\mathcal{U}}_{tD}$ introduced in previous section. We prove that they are self-adjoint and negative-definite. 
We also find under some condition a common core  for the generators of the semigroups $\mathcal{U}_{tD_1}$ and $\mathcal{U}_{tD_2}$, $\widehat{\mathcal{U}}_{tD_1}$ and $\widehat{\mathcal{U}}_{tD_2}$, and finally, $\mathcal{U}_{tD_1}$ and $\widehat{\mathcal{U}}_{tD_2}$ with $D_1 \ne D_2$. And, most importantly, we prove that closure of positive linear combinations of generators of averaging semigroup in the coordinate and momentum representations are self-adjoint and under some conditions generate strongly continuous semigroups. This allows us to define a Hamiltonian of an infinite-dimensional quantum oscillator.

\begin{definition}\label{def-smooth-functions}
    Let $D$ be a nuclear positive-definite self-adjoint operator whose eigenvectors coincide with the basis vectors ${\mathcal E}$. Denote by $\widehat{C}_{D}^\infty$ the set ${\rm span}\{\widehat{\mathcal{U}}_{t D} {\cal H}_{\cal E},\, t>0\}$. If additionally $D^{1/2}$ is nuclear, then we denote by $C_{D}^\infty$ the set ${\rm span}\{\mathcal{U}_{t D}{\cal H}_{\cal E},\, t>0\}$. Both of these subspaces are linear subspaces in ${\cal H}_{\cal E}$. We call $C_D^\infty$ as the space of smooth functions, because (as we see later) its elements have mixed derivatives of any order along the coordinate directions.
\end{definition}

\begin{lemma}\label{lemma-hat-c-dense}
    Let $D$ be a nuclear positive-definite self-adjoint operator whose eigenvectors coincide with the basis vectors ${\mathcal E}$.
    Then the linear subspace $\widehat{C}_{D}^\infty$ is dense in ${\cal H}_{\cal E}$.
\end{lemma}
 
\begin{proof}
    The statement follows from the strong continuity of the semigroup established in the corollary~\ref{cor-mom-smooth-cont}. For any function $f \in \mathcal{H}_\mathcal{E}$ and any $t>0$, the function $\widehat{\mathcal{U}}_{tD} f$ belongs to $\widehat{C}_D^\infty$, and $\lim\limits_{t\to 0} \|\widehat{\mathcal{U}}_{tD} f - f\|_{\mathcal{H}_\mathcal{E}} = 0$.
\end{proof}

\begin{lemma}\label{lemma-c-dense}
    Let $D$ be a nuclear positive-definite self-adjoint operator whose eigenvectors coincide with the basis vectors ${\mathcal E}$. If $D^{1/2}$ is nuclear, then $C_{D}^\infty$ is tight in ${\cal H}_{\cal E}$.
\end{lemma}

\begin{proof}
    The proof is similar to lemma \ref{lemma-hat-c-dense}. The continuity of the semigroup $\mathcal{U}_{tD}$ is established as a consequence of corollary \ref{cor-coord-smooth-cont}.
\end{proof} 

\begin{lemma}\label{lemma-C-subset}
    Let $D$ and $B$ be nuclear positive-definite self-adjoint operators whose eigenvectors coincide with the basis vectors ${\mathcal E}$. Let also $D\prec C B$ for some $C>0$, that is, the inequality $d_k\le C b_k$ is satisfied for the corresponding eigenvalues. Then $\widehat{C}_B^\infty\subseteq\widehat{C}_D^\infty$, and $\widehat{C}_B^\infty$ is invariant with respect to the action of $\widehat{\mathcal{U}}_{tD}$. If additionally $D^{1/2}$ is nuclear, then $C_B^\infty\subseteq C_D^\infty$, and $C_B^\infty$ is invariant with respect to the action of $\mathcal{U}_{tD}$.
\end{lemma}
\begin{proof}
    Note that the operator $CB - D$ is also nuclear, positive definite, and its eigenvectors coincide with the basis vectors ${\mathcal E}$.
    
    Let $f\in C_B^\infty$. Then $f =\mathcal{U}_{tB}u$ for some $t > 0$ and $u\in\mathcal{H}_\mathcal{E}$, and
    $$
    f = \mathcal{U}_{(t/C) CB} u = \mathcal{U}_{(t/C) D} \mathcal{U}_{(t/C) (CB -D)} u \in C_D^\infty.
    $$
    So,
    $$
    \mathcal{U}_{sD} f = \mathcal{U}_{sD} \mathcal{U}_{tB} u = \mathcal{U}_{sD + tB} u = \mathcal{U}_{tB}  \mathcal{U}_{sD} u \in C_B^\infty.
    $$
    The inclusion $\widehat{C}_B^\infty\subseteq\widehat{C}_D^\infty$ and the invariance of $\widehat{C}_B^\infty$ with respect to $\widehat{\mathcal{U}}_{tD}$ are proved similarly.
\end{proof}

\begin{lemma}\label{lemma-x-exp}
     Let $D$ and $B$ be nuclear positive-definite self-adjoint operators whose eigenvectors coincide with the basis vectors ${\mathcal E}$. 
    Then the multiplication operator of a function from  $\widehat{C}_D^\infty$ by $x_j$ is well-defined and
    $$
    \| x_j\ \widehat{\mathcal{U}}_{tD} f\|_{{\cal H}_{\cal E}}\leq \frac{1}{\sqrt {etd_j}} \|f\|_{{\cal H}_{\cal E}}.
    $$
\end{lemma}

\begin{proof}
    Consider on the space $L_2(\mathbb{R})$ the multiplication operator by a function \\
    $x\exp \left(-\dfrac{1}{2} t d_j x^2\right)$, whose norm in $L_2(\mathbb{R})$ is equal to
    $$
    \sup\limits_{x\in {\mathbb R}} 
    \left| x \exp \left(-\dfrac{1}{2} t d_j x^2\right) \right| 
    = \dfrac{1}{\sqrt{etd_j}}.
    $$ 
    Since
    $$
    {\cal H}_{\cal E}=L_2(\mathbb{R})\otimes {\cal H}_{e_1,\ldots.,e_{j-1},e_{j+1},\ldots},
    $$
    the operator $f\mapsto x_j \ \widehat{\mathcal{U}}_{tD} f$ is represented as a tensor product of operators
    $$ 
    x_j \ \widehat{\mathcal{U}}_{tD}
    =\left( x_j\exp \left(-\dfrac{1}{2} t d_j x^2\right)\right)
    \otimes \exp \left(-\dfrac{t}{2}\sum\limits_{k\neq j}  d_k x_k^2)\right),
    $$
    where the second operator is a smoothing in the momentum representation, and therefore its norm does not exceed one. Therefore, $\| x_j \ \widehat{\mathcal{U}}_{tD}\|_{B({\cal H}_{\cal E})}\leq \frac{1}{\sqrt {etd_j}}$. 
\end{proof}

\begin{corollary}\label{cor-xl-exp}
    Let $D$ be a nuclear positive-definite self-adjoint operator whose eigenvectors coincide with the basis vectors ${\mathcal E}$.
    For any $l\in\mathbb N$, the product of the function $\widehat{\mathcal{U}}_{tD} f\in\widehat{C}_D^\infty$ by $x_{j_1}\ldots x_{j_l}$ is well defined, and the following estimate is correct
    $$ 
    \left\|x_{j_1} \ldots x_{j_l} \ \widehat{\mathcal{U}}_{tD} f \right\|_{\mathcal{H}_\mathcal{E}}
    \le \left( \dfrac{l}{et} \right)^{l/2} \dfrac{1}{\sqrt{d_1 \ldots d_l}} \|f\| _{\mathcal{H}_\mathcal{E}}.
    $$
    In particular,
    $$ 
    \left\|x_j^l \ \widehat{\mathcal{U}}_{tD} f \right\|_{\mathcal{H}_\mathcal{E}}
    \le \left( \dfrac{l}{etd_j} \right)^{l/2}  \|f\| _{\mathcal{H}_\mathcal{E}}.
    $$
\end{corollary}

 \begin{proof}
    The estimation follows from the fact that for any $l\in\mathbb N$ the multiplication operator by a function $x_{j_1} \ldots x_{j_l}\ \exp \left(-\dfrac{1}{2} t d_j x^2\right)$ in $L_2({\mathbb R})$ can be represented as a product of $l$ operators of the form $x_{j_k} \  \widehat{\mathcal{U}}_{(t/l) D}$.
 \end{proof}

\begin{lemma}\label{lemma-der-ineq}
    Let $D$ be a nuclear positive-definite self-adjoint operator whose eigenvectors coincide with the basis vectors ${\mathcal E}$. If $D^{1/2}$ is nuclear, then any smooth function $\mathcal{U}_{tD} f\in C^\infty_D$ has any higher-order mixed derivative $\partial_{j_1}\ldots\partial_{j_k}$ along the basic directions, which admits evaluation
    $$ 
    \left\|\partial_{j_1} \ldots \partial_{j_k} \ \mathcal{U}_{tD} f \right\|_{\mathcal{H}_\mathcal{E}}
    \le \left( \dfrac{l}{et} \right)^{l/2} \dfrac{1}{\sqrt{d_{j_1} \ldots d_{j_k}}} \|f\| _{\mathcal{H}_\mathcal{E}}.
    $$
    In particular,
    $$ 
    \left\|\partial_j^k \ \mathcal{U}_{tD} f \right\|_{\mathcal{H}_\mathcal{E}}
    \le \left( \dfrac{l}{etd_j} \right)^{l/2}  \|f\| _{\mathcal{H}_\mathcal{E}}.
    $$
\end{lemma}

\begin{proof}
    Let us take the maximum $j_k$ among all indexes and denote it as $J$.  Applying to the function $\partial_{j_1}\ldots\partial_{j_k} \mathcal{U}_{tD} f$ the Fourier transform along the first $J$ basic directions gives
    $$
    \left\|\partial_{j_1} \ldots \partial_{j_k} \ \mathcal{U}_{tD} f \right\|_{\mathcal{H}_\mathcal{E}}
    = \left\|(x_{j_1} \ldots x_{j_l} \ \widehat{\mathcal{U}}_{tD_J}) \otimes \mathcal{U}_{t(D-D_J)} f \right\|_{\mathcal{H}_\mathcal{E}}.
    $$
    The norm of the operator $(x_{j_1} \ldots x_{j_l} \ \widehat{\mathcal{U}}_{tD_J})$ does not exceed $\left(\dfrac{l}{et}\right)^{l/2}\dfrac{1}{\sqrt{d_{j_1} \ldots d_{j_k}}}$, and the norm of $\mathcal{U}_{t(D-D_J)}$ does not exceed one.
\end{proof}

\begin{lemma}\label{lemma-LD}
    Let $D$ and $B$ be non-negative nuclear operators on the space $E$ whose eigenvectors coincide with the basis vectors ${\mathcal E}$. If the operator $D B^{-1}$ is nuclear, then the multiplication operator $L_{D,B}(t)$ of the function $f\in {\cal H}_{\cal E}$ by  $(Dx,x)\exp (-t(B x,x))$ is well defined and satisfies
    $$
    \|L_{D,B}(t)\|_{B({\cal H}_{\cal E})}\leq \dfrac{1}{et}{\rm Tr}(D B^{-1}).
    $$ 
\end{lemma} 
\begin{proof}
    Indeed, $L_{D,B}(t)=\sum\limits_{j=1}^{\infty}d_j x_j^2\exp (-t(Bx,x))$. Therefore, by virtue of the corollary \ref{cor-xl-exp} the inequality holds
    $$
    \| L_{D,B}(t)\|_{B({\cal H}_{\cal E})} \leq \sum\limits_{j=1}^{\infty }d_j \dfrac{1}{etb_j}=\dfrac{1}{et}{\rm Tr}(D B^{-1}).
    $$
\end{proof}

\begin{lemma}\label{lemma-Dxx-generator}
    Let $D$ and $B$ be non-negative nuclear operators on the space $E$ whose eigenvectors coincide with the basis vectors ${\mathcal E}$. If the operator $D B^{-1}$ is nuclear, then the linear subspace $\widehat{C}_{B}^\infty$ lies in the domain of the generator of the strongly continuous semigroup $\widehat{\mathcal{U}}_{t D},\, t\geq 0$, and
    $$
    \dfrac{\d}{\d t} (\widehat{\mathcal{U}}_{t D} f)|_{t=0}=-(Dx,x)f \quad \forall f\in \widehat{C}_{B}^\infty.
    $$
\end{lemma}

\begin{proof}
    Indeed, let $f \in \widehat{C}_B^\infty$. Therefore, there exists $s>0$ and $u\in {\cal H}_{\cal E}$ such that $f = \widehat{\mathcal{U}}_{s B} u$. Then
    $$
    \dfrac{\d}{\d t} (\widehat{\mathcal{U}}_{t D} f)|_{t=0}
    =\dfrac{\d}{\d t} (\widehat{\mathcal{U}}_{t D} \widehat{\mathcal{U}}_{s B} u)|_{t=0}
    =\lim\limits_{t\to 0}
    \left[\dfrac{1}{t} (\widehat{\mathcal{U}}_{t D} - I)\widehat{\mathcal{U}}_{s B} u
    \right].
    $$
    Mapping ${\mathbb R}_+\to {\cal H}_{\cal E}$ which acts as $t\mapsto\widehat{\mathcal{U}}_{t D} \widehat{\mathcal{U}}_{s B} u$ is the limit at $N\to\infty$ in the space $C({\mathbb R}_+, {\cal H}_{\cal E})$ of sequence of Frechet-differentiable maps $t\mapsto\widehat{\mathcal{U}}_{t D_N}\widehat{\mathcal{U}}_{s B} u$. For every $N\in\mathbb N$, one has
    $$
    \dfrac{\d}{\d t} \widehat{\mathcal{U}}_{t D_N} \widehat{\mathcal{U}}_{s B} u
    = -(D_N x,x)\widehat{\mathcal{U}}_{t D_N} \widehat{\mathcal{U}}_{s B},\quad t\geq 0.
    $$ 
    Therefore, by virtue of lemma \ref{lemma-LD}, the functional sequence $\{ -(D_N x,x)\widehat{\mathcal{U}}_{t D_N} \widehat{\mathcal{U}}_{s B} u,\ t\geq 0\, \} $ converges in the space $C({\mathbb R}_+, {\cal H}_{\cal E})$ to the function
    $$
    -(Dx,x)(\widehat{\mathcal{U}}_{t D} \widehat{\mathcal{U}}_{s B} u) 
    =-L_{D,B}(s/2) \widehat{\mathcal{U}}_{(s/2) B} \widehat{\mathcal{U}}_{t D}  u,\ t\geq 0.
    $$
    Therefore, the function $t\to \widehat{\mathcal{U}}_{t D} \widehat{\mathcal{U}}_{s B} u$ is Frechet differentiable, and
    $$
    \dfrac{\d}{\d t} (\widehat{\mathcal{U}}_{t D} f)|_{t=0}
    = L_{D,B}(s/ 2) \widehat{\mathcal{U}}_{(s/2) B} u=-(Dx,x)f.
    $$
\end{proof}

\begin{theorem}\label{th-lambda-core}
    Let $D$ and $B$ be non-negative nuclear operators on the space $E$ whose eigenvectors coincide with the basis vectors ${\mathcal E}$. If the operator $D B^{-1}$ is nuclear, then the generator of the semigroup $\widehat{\mathcal{U}}_D(t)$ is the closure of the semi-bounded operator $\Lambda_D: \widehat{C}_B^\infty\to\mathcal{H}_\mathcal{E}$, acting as
    $$
    \Lambda_D f = - (Dx,x) f, \quad f \in \widehat{C}_B.
    $$
\end{theorem}

\begin{proof}
    By virtue of lemma \ref{lemma-LD}, $\Lambda_D:\widehat{C}_B\to\mathcal{H}_\mathcal{E}$ is densely defined, symmetric, and non-positive. Therefore, the conjugate operator is densely defined and the operator $\Lambda_D$ is closed. Hence we only need to show that the domain of the closure of the operator $\Lambda_D:\widehat{C}_B^\infty\to\mathcal{H}_\mathcal{E}$ coincides with the domain  of the generator of the semigroup $\widehat{\mathcal{U}}_D(t)$, that is, that $\widehat{C}_B^\infty$ is dense in the domain of the generator with respect to the norm $\|f\|_{\Lambda_D} = \|f\|_{\mathcal{H}_\mathcal{E}} + \|\Lambda_D f\|_{\mathcal{H}_\mathcal{E}}$.

    If $DB^{-1}$ is nuclear, then the ratio of eigenvalues of $D$ and $B$, $\dfrac{d_k}{b_k}$, tends to zero, and consequently, there exists such a $C > 0$ that $d_k \le C b_k$. Therefore, $D\prec C B$, and by lemma \ref{lemma-C-subset} the space $\widehat{C}_B^\infty$ is invariant with respect to the action of the semigroup $\widehat{\mathcal{U}}_D(t)$. In addition, $\widehat{C}_B^\infty$ is dense in $\mathcal{H}_\mathcal{E}$ by lemma~\ref{lemma-hat-c-dense}. According to Proposition 7 of \cite{engel-nagel}, it 
    implies that the subspace $\widehat{C}_B^\infty$ is an essential domain for
    the generator of the semigroup $\widehat{\mathcal{U}}_D(t)$, that is, $\widehat{C}_B^\infty$ is dense in the domain of the generator 
    with respect to the 
    $\Lambda_D$-norm
\end{proof}

\begin{theorem}\label{th-delta-core}
    Let $D$ and $B$ be non-negative nuclear operators on the space $E$ whose eigenvectors coincide with the basis vectors ${\mathcal E}$. If the operators $D B^{-1}$, $D^{1/2}$, $B^{1/2}$ are nuclear, then the generator of the semigroup $\mathcal{U}_D(t)$ is the closure of the semi-bounded operator $\Delta_D:C_B^\infty\to\mathcal{H}_\mathcal{E}$, acting as
    $$
    \Delta_D f = \sum\limits_{k=1}^\infty d_k \partial_k^2 f, \quad f \in C_B.
    $$
 \end{theorem}
 \begin{proof}
    The proof is almost completely analogous to the proof of lemmas \ref{lemma-LD} and \ref{lemma-Dxx-generator} and of the theorem \ref{th-lambda-core}.
 \end{proof}

 \begin{remark}
    For example, the case of $B=D^{1/2}$ fits the condition of the theorem \ref{th-delta-core} if $D^{1/2}$ is nuclear. The case $B =D^{2/3}$  fits the condition of the theorem \ref{th-lambda-core} if $D^{1/3}$ is nuclear. Thus, operators $B$ satisfying the condition of the theorem exist= for a wide class of nuclear operators $D$.
\end{remark}

\begin{theorem}\label{th-sum-core}
    Let $D_x$, $D_p$ and $B$ be non-negative nuclear operators on the space $E$ whose eigenvectors coincide with the basis vectors ${\mathcal E}$. If the operators $D_x^{1/2}$, $B^{1/2}$, $D_x B^{-1}$ and $D_p B^{-1}$ are nuclear, then for any positive numbers $a,b\in\mathbb{R}$ the operator $a\Delta_{D_x} + b\Lambda_{D_p}$ is correctly defined on the space
    $$
    D_B 
    = \{\widehat{\mathcal{U}}_{sB} \mathcal{U}_{tB}  u \mid u \in \mathcal{H}_\mathcal{E}, t > 0, s>0\}, 
    $$
    its closure is a self-adjoint operator and is a generator of some semigroup.
\end{theorem}

\begin{proof}
    Since $D_B \subset \widehat{C}_B^\infty$, operator $\Lambda_D$ is well defined on $D_B$ by the theorem \ref{th-delta-core}. Now let us calculate the derivatives of an arbitrary function $f = \exp(-s(B x,x)) \mathcal{U}_{t B}  u \in D_B$:
\begin{eqnarray*}
    \partial_k f &=& \partial_k \exp(-s(B x,x)) \mathcal{U}_{t B}  u =  - 2 s b_k x_k \exp(-s(B x,x)) \mathcal{U}_{t B}  u_2 + \exp(-s(B x,x)) \partial_k \mathcal{U}_{t B}  u, \\
    \partial_k^2 f_2 
    &=& 4 s^2 b_k^2 x_k^2 \exp(-s(B x,x)) \mathcal{U}_{t B}  u \\
    &=&   - 4 s b_k x_k \exp(-s(B x,x)) \partial_k \mathcal{U}_{t B}  u
    + \exp(-s(B x,x)) \partial_k^2 \mathcal{U}_{t B}  u.
\end{eqnarray*}
    Let us estimate the norm of each term of the second derivative:
\begin{equation}\label{DB-der-norm-1}
        \|4 s_2^2 b_k^2 x_k^2 \exp(-s_2(B x,x)) \mathcal{U}_{t_2 B}  u_2 \| _{\mathcal{H}_\mathcal{E}}
        \le 4 s_2^2 b_k^2  \dfrac{1}{es_2 b_k} \|\mathcal{U}_{t_2 B}  u_2\| _{\mathcal{H}_\mathcal{E}}
        \le C_1 b_k \|u_2\| _{\mathcal{H}_\mathcal{E}},
    \end{equation}
    \begin{equation}\label{DB-der-norm-2}
        \begin{gathered}
            \| - 4 s_2 b_k x_k \exp(-s_2(B x,x)) \partial_k \mathcal{U}_{t_2 B}  u_2 \| _{\mathcal{H}_\mathcal{E}}
            \le 
            4 s_2 b_k \dfrac{1}{\sqrt{e s_2 b_k}} 
            \| \partial_k \mathcal{U}_{t_2 B}  u_2 \| _{\mathcal{H}_\mathcal{E}} \le \\
            \le 4 s_2 b_k \dfrac{1}{e \sqrt{s_2 t_2} b_k} \|u_2\|  _{\mathcal{H}_\mathcal{E}} = C_2 \|u_2\| _{\mathcal{H}_\mathcal{E}},
        \end{gathered}
    \end{equation}
    \begin{equation}\label{DB-der-norm-3}
        \| \exp(-s_2(B x,x)) \partial_k^2 \mathcal{U}_{t_2 B}  u_2 \| _{\mathcal{H}_\mathcal{E}} 
        \le 
        \| \partial_k^2 \mathcal{U}_{t_2 B}  u_2 \| _{\mathcal{H}_\mathcal{E}}
        \le 
        \dfrac{1}{e s_t b_k} \|u_2\| _{\mathcal{H}_\mathcal{E}} = \dfrac{C_3}{b_k} \|u_2\| _{\mathcal{H}_\mathcal{E}}.
    \end{equation}
    As one can see from the estimates (\ref{DB-der-norm-1})-(\ref{DB-der-norm-3}), the series $\sum\limits_{k=1}^\infty d_{x,k}\partial_k^2 f_2$ converges:
    \begin{equation}\label{DB-der-norm}
        \begin{gathered}
            \sum\limits_{k=1}^\infty d_{x,k} \|\partial_k^2 f_2 \| _{\mathcal{H}_\mathcal{E}} 
            \le
            \sum\limits_{k=1}^\infty d_{x,k} \left(C_1 b_k + C_2 + \dfrac{C_3}{b_k} \right) \|u\| _{\mathcal{H}_\mathcal{E}} 
            \le \\ 
            \le
         (C_1 {\rm Tr}(D_x B) + C_2 {\rm Tr}(D_x) + C_3 {\rm Tr}(D_x B^{-1}) \|u\| _{\mathcal{H}_\mathcal{E}}.
        \end{gathered}
    \end{equation}
    Thus, the operator $\Delta_{D_x}$ is correctly defined on the space $D_B$.

    The subspace $D_B$ is dense in $\mathcal{H}_\mathcal{E}$ due to the continuity of the semigroups $\mathcal{U}_{tB}$ and $\widehat{\mathcal{U}}_{tB}$.

    Both operators $\Delta_{D_x}$ and $\Lambda_{D_p}$ are self-adjoint and positive-definite, so their positive linear combination is symmetric and positive-definite. Consider three sesquilinear forms:
\begin{eqnarray*}
    t_1(f,g) &=& ((I - a \Delta_{D_x}) f, g), \quad D(t_1) = D_B,\\
    t_2(f,g) &=& ((I + b (D_p x, x)) f, g), \quad D(t_1) = D_B,\\
    t_3(f,g) &=& ((I - a \Delta_{D_x} + b (D_p x, x) f, g), \quad D(t_2) = D_B.
\end{eqnarray*}
    All these forms are closable, since the operators $I - a\Delta_{D_x}$, $I+b (D_p x, x)$ and $I - a\Delta_{D_x}+b (D_p x,x)$ are symmetric and semi-bounded from below. The forms $\bar t_1$ and $\bar t_2$ are represented exactly by the self-adjoint operators $I - a\Delta_{D_x}$ and $I + b (D_p x, x)$, and the form $\bar t_3$ is represented by some self-adjoint operator $(I - a\Delta_{D_x} - b\Lambda_{D_p})_{\rm f}$, which is Friedrichs extension of the operator $I - a\Delta_{D_x} - b\Lambda_{D_p}$. Since inequalities
    $\bar t_1 \le \bar t_3$
    and
    $\bar t_2 \le \bar t_3$
    are satisfied,
    then, by the theorem \ref{th-repr-operators-incl}, the domain of $(I - a\Delta_{D_x} - b\Lambda_{D_p})_{\rm f}$ is a subset of $D(I - a \Delta_{D_x}) = D(\Delta_{D_x})$, and a subset of $D(I - b(D_p x, x)) = D((D_p x, x))$.

    Let us show that $(I - a\Delta_{D_x} - b\Lambda_{D_p})_{\rm f}$ coincides with the closure of the operator $I - a\Delta_{D_x} - b\Lambda_{D_p}$. Consider an arbitrary $f$ from the domain of the Friedrichs extension $(I - a\Delta_{D_x} - b\Lambda_{D_p})_{\rm f}$. Then $f$ belongs to $D(\Delta_{D_x}) \cap D(\Lambda_{D_p})$. This means that $f$ is twice differentiable along each basic direction, multiplication by $x_k^2$ is correctly defined for it, and that there exists a number $N_1$ such that
    $$
    \sum\limits_{k=N_1+1}^\infty d_{x,k} \|\partial_k^2 f\| + d_{p,k} \|x_k^2 f\| < \varepsilon.
    $$
    On the other hand, for $f_t =\mathcal{U}_{tB}\widehat{\mathcal{U}}_{tB} f \in D_B$, according to the inequality \ref{DB-der-norm}, there exists a number $N_2$ such that
    $$
    \sum\limits_{k=N_2+1}^\infty d_{x,k} \|\partial_k^2 f_t\| + d_{p,k} \|x_k^2 f_t\| < \varepsilon.
    $$
    Thus, if $N = \max\{N_1, N_2\}$, then
    $$
    \| (a \Delta_{D_x} + b \Lambda_{D_p}) (f_t - f)\| _{\mathcal{H}_\mathcal{E}}
    \le \left\| 
        \sum\limits_{k=1}^N (a \cdot d_{p,k} x_k^2 + b \cdot d_{x,k} \partial_k^2)(f_t - f)
    \right\| + 2 \varepsilon
    $$
    and in the finite-dimensional case, the specified quantity tends to zero.
\end{proof}

\section{On the Fourier transform on the space \texorpdfstring{${\cal H}_{\cal E}$}{TEXT}}\label{Sec:Fourier}

In this section, we investigate the question of the existence of the Fourier transform on the space ${\cal H}_{\cal E}$. Particularly, we prove that sequence of partial Fourier transforms over first $n$ coordinates converges to $0$ in strong topology at $n\to\infty$. Also we prove that semigroups $\mathcal{U}_{tD}$ and $\widehat{\mathcal{U}}_{tD}$ of averaging shifts in the coordinate and momentum representations are not unitary equivalent. Thus, there is no such a unitary transformation of $E$ onto itself which maps any shift $S_{th}$ in the coordinate representation into shift $\widehat{S}_{th}$ in the momentum representation.

Let $\nu_D$ be a centered Gaussian measure on the Hilbert space $E$ having the nuclear covariance operator $D$, and let $\nu_{tD},\,t\geq 0,$ be a semigroup with respect to the convolution operation of centered Gaussian measures with covariance operators $tD,\, t\geq 0$.

In the case of a finite-dimensional Euclidean space $E$, the Fourier transform on the space $L_2(E)$ implements the unitary equivalence of the shift operator $S_h$ on the vector $h$ in the coordinate space to the shift operator $\widehat{S}_h$ on the same vector in momentum space, as well as the unitary equivalence of the semigroup $\mathcal{U}_{tD} = \int\limits_E S_{h} d\nu _{tD}(h),\, t\geq 0$ of diffusion operators to the semigroup $\widehat{\mathcal{U}}_{tD},\, t\geq 0$ of multiplication operators by Gaussian exponents.

If the space $E$ is infinite-dimensional, then from the analysis of the properties of the semigroup $\widehat{\mathcal{U}}_{tD},\, t\geq 0$, carried out in the theorem \ref{th-mom-smooth-appr}, and the analysis of the properties of the semigroup $\mathcal{U}_{tD},\, t\geq 0$ in the~\cite{SZ19} follow the absence on the space ${\cal H}_{\cal E}$ of the Fourier transform, which has all its characteristic properties. Namely, the following statement is true.

\begin{theorem}
    Let $E$ be an infinite-dimensional real separable Hilbert space and $\cal E$ be some ONB in this space. Then there does not exist a unitary transformation of the space ${\cal H}_{\cal E}$ onto itself, which, when arbitrarily choosing a non-negative nuclear operator $D$ with the basis of eigenvectors $\cal E$, implements the unitary equivalence of the semigroup $\mathcal{U}_{tD},\, t\geq 0$ and the semigroups $\widehat{\mathcal{U}}_{tD},\,t\geq 0.$
\end{theorem}

\begin{proof}
    Suppose, to the contrary, that there exists a unitary transformation ${\cal F}$ of the space ${\cal H}_{\cal E}$ onto itself such that, when arbitrarily choosing a non-negative nuclear operator $D$ with the basis of eigenvectors ${\cal E}$, the following equalities are satisfied:
    \begin{equation}\label{F}
        \mathcal{U}_{tD}
        ={\cal F} \widehat{\mathcal{U}}_{tD} 
        {\cal F}^{-1},\quad t\geq 0.
    \end{equation}
    Consider a non-negative nuclear operator $D$ such that the operator $D^{1/2}$ is not nuclear. Then, by virtue of the corollary \ref{cor-mom-smooth-cont}, the semigroup $\widehat{\mathcal{U}}_{tD},\, t\geq 0,$ is strongly continuous in the space ${\cal H}_{\cal E}$, and the semigroup $\mathcal{U}_{tD},\, t\geq0,$ does not have the property of strong continuity due to the theorem \ref{th-coord-smooth-trivial}. But this contradicts equality (\ref{F}).
\end{proof}

\begin{remark}
    The absence on the space ${\cal H}_{\cal E}$ of a unitary Fourier transform having the property (\ref{F}) does not exclude the possibility of the existence on the space ${\cal H}$ of such a unitary transformation that implements the unitary equivalence of the shift operator $S_h$ to the vector $h$ in the coordinate space and the operator $\widehat{S}_h$ shift by the same vector in the momentum space, but does not contain the space ${\cal H}_{\cal E}$ as an invariant subspace.
\end{remark}

On the space ${\cal H}_{\cal E}$, partial Fourier transforms ${\cal F}_{n},\, n\in {\mathbb N}$ were defined, representing the Fourier transform over the first $n$ variables. For every $n\in\mathbb N$, the partial Fourier transform admits the representation ${\cal F}_n={\cal F}_{E_n}\otimes I_{E^n}$, where ${\cal F}_{E_n}$ is the Fourier transform on the space $L_2(E_n)$ and $I_{E^n}$ is the identical operator on the space ${\cal H}_{{\cal E}^n}$, associated with the representation of the space ${\cal H}_{\cal E}$ as a tensor product $L_2(E_n)\otimes {\cal H}_{{\cal E}^n}$. Therefore, partial Fourier transforms are unitary operators on the space ${\cal H}_{\cal E}$.

If $B$ is a nonnegative operator of finite rank $m\in\mathbb N$, then the equality holds
$$
\widehat{\mathcal{U}}_D(t)={\cal F}_m \mathcal{U}_{D}(t){\cal F}_m^{-1},\quad t\geq 0.
$$
Later for studying the Fourier transform on the space $\cal H$, we need to investigate whether the functions $f: E\to {\mathbb C}$, represented by infinite products of functions from one variable of the form $f(x)=\prod\limits_{k=1}^{\infty}f_k(x_k)$ belong to the space ${\cal H}_{\cal E}$.

\begin{definition}
    We say that an infinite product 
    $\prod\limits_{i=1}^{\infty}f_i(x_i)$ is defined as an element of ${\cal H}_{\cal E}$ if there exists a function $g\in {\cal H}_{\cal E}$, which, with an arbitrary choice of an absolutely measurable bar $\Pi_{a,b}$, satisfies the equality
    \begin{equation}\label{prod}
    (g,\chi _{\Pi _{a,b}})=\prod\limits_{i=1}^{\infty }\int\limits_{a_i}^{b_i}f_i(x_i)dx_i.
    \end{equation}
\end{definition}

To investigate the existence of infinite products of functions, the following statement is required:

\begin{lemma}\label{lemma-inf-prod-ineq}
    Let $\{{u}_k\}, \{{v}_k\}$ be sequences of vectors of a Banach space $\cal K$ such that $\{ \|{u}_k\|_{\cal K}\}$,  $\{ \|{v}_k\|_{\cal K}\}\in l_{\infty }$ and 
    \begin{equation}\label{crit}
        \sum\limits_{k=1}^{\infty }\max \{0, \ln (\|u_k\|_{\cal K})\} \leq \infty ,\ \sum\limits_{k=1}^{\infty }\max \{ 0,\ln (\|v_k\|_{\cal K})\} \leq \infty .
    \end{equation}
    Then for any $m\in {\mathbb N}$ we can evaluate
    \begin{equation}\label{crit'}
        \| \otimes_{k=1}^m v_k - \otimes_{k=1}^m u_k \|_{{\cal K}^m} \leq C^2\sum\limits_{k=1}^{\infty }\|v_k-u_k\|_{\cal K}, 
    \end{equation}
    where
    $$
    C=\max \{ 1, \sup\limits_{K_1\leq K_2}\prod\limits_{k=K_1}^{K_2}\|u_k\|_{\cal K} , \sup\limits_{K_1\leq K_2}\prod\limits_{k=K_1}^{K_2}\|v_k\|_{\cal K}\}
    $$
    and ${\cal K}^m={\cal K}\otimes \ldots\otimes {\cal K}$ is tensor product of the space $\cal K$ on itself $m$ times, endowed with norm $\| \cdot \|_{{\cal K}^m}$ such that $\| g_1\otimes \ldots \otimes g_m\|_{{\cal K}^m}=\prod\limits_{k=1}^m\| g_k\|_{\cal K}$ for any $g_1,\ldots,g_m\in \cal K$.
\end{lemma}

\begin{proof}
    Due to the conditions (\ref{crit}), infinite products of $\sum\limits_{k=1}^{\infty }\|u_k\|_{\cal K}$ and $\sum\limits_{k=1}^{\infty}\|v_k\|_{\cal K}$ converge unconditionally and therefore the value of $C$ introduced in the formulation of lemma \ref{lemma-inf-prod-ineq} is finite. Next, for every $m\in\mathbb N$ we have
    $$
    {\|} \otimes_{k=1}^m v_k-\otimes_{k=1}^m u_k {\|}_{{\cal K}^m}=\| \otimes_{k=1}^m v_k - v_1\otimes (\otimes_{k=2}^mu_k)+  v_1\otimes (\otimes_{k=2}^mu_k)-
    \otimes_{k=1}^mu_k\|_{{\cal K}^m} \leq
    $$
    $$
    \|v_1\|_{\cal K}\, \| \otimes_{k=2}^mv_k-\otimes_{k=2}^mu_k\|_{{\cal K}^{m-1}}+\|v_1-u_1\|_{\cal K}\, \|\otimes_{k=2}^mu_k\|_{{\cal K}^{m-1}}=
    $$
    $$
    = \|v_1\|_{\cal K}\, \| \otimes_{k=2}^mv_k-\otimes_{k=2}^mu_k\|_{{\cal K}^{m-1}}+\|v_1-u_1\|_{\cal K}\, \prod\limits_{k=2}^m\|u_k\|_{{\cal K}} \leq $$
    $$\leq \|v_1\|_{\cal K}  \|\otimes_{k=2}^mv_k-\otimes_{k=2}^mu_k\|_{{\cal K}^{m-1}}+C\|v_1-u_1\|_{\cal K}.
    $$
    Similarly, $\|\otimes_{k=2}^mv_k-\otimes_{k=2}^mu_k\|_{{\cal K}^{m-1}}\leq C\|v_2-u_2\|_{\cal K}+\|v_2\|_{\cal K}\,  \|\otimes_{k=3}^mv_k-\otimes_{k=3}^mu_k\|_{{\cal K}^{m-2}}$, and so on. Finally we get an estimate 
    $$
        {\|} \otimes_{k=1}^m v_k-\otimes_{k=1}^m u_k {\|}_{{\cal K}^m}
    \leq 
    $$ 
    $$
        C\|v_1-u_1\|_{\cal K}+\|v_1\|_{\cal K}C\|v_2-u_2\|_{\cal K}+\|v_1\|_{\cal K}\|v_2\|_{\cal K}C\|v_3-u_3\|_{\cal K}+\ldots+\|v_1\|_{\cal K}\ldots\|v_{m-1}\|_{\cal K}C\|v_m-u_m\|_{\cal K}\leq
    $$
    $$
        \leq \sum\limits_{k=1}^mC^2\|v_k-u_k\|_{\cal K}\leq \sum\limits_{k=1}^{\infty }C^2\|v_k-u_k\|_{\cal K}.
    $$

Therefore, for all $m\in\mathbb N$, the estimate (\ref{crit'}) holds.
\end{proof}

The following theorem provides a sufficient condition for the existence of an infinite product.

\begin{theorem}
    Consider a sequence of functions $\{f_i\}$ from the space $L_2({\mathbb R})$ such that there exists an absolutely measurable block $\Pi _{a,b}$ such that
    $$
    \sum\limits_{i=1}^{\infty }\| f_i-\chi _{[a_i,b_i]}\|_{L_2({\mathbb R})}<+\infty.
    $$
    Then there exists an infinite product
    $\prod\limits_{i=1}^{\infty }f_i(x_i)\in {\cal H}_{\cal E}$.
\end{theorem}

\begin{proof}
    Define $\delta _i=\| f_i-\chi _{[a_i,b_i]}\|_{L_2(\mathbb R)}$ for $i\in \mathbb{Z}_+$. By condition, the series $\sum\limits_{j=1}^{\infty }\delta _j$ converges. The product $\prod\limits_{k=1}^{\infty }\|\chi _{[a_k,b_k]}\|_{L_2({\mathbb R})}$ converges absolutely since the block $\Pi _{a,b}$ is absolutely measurable. Therefore, the product $\prod\limits_{k=1}^{\infty }\| f _k\|_{L_2({\mathbb R})}$ converges absolutely, since 
    $$
    \| \chi _{[a_k,b_k]}\|_{L_2({\mathbb R})}-\delta _k\leq \| f_k\|_{L_2({\mathbb R})}]\leq \| \chi _{[a_k,b_k]}\|_{L_2({\mathbb R})}+\delta _k.
    $$
    For each number $n\in\mathbb N$, we choose such a sequence of non-negative numbers $\varepsilon _{n,i},i\in\mathbb N$ that $\sum\limits_{i=1}^{\infty}\varepsilon _{n,i}<+\infty$ and, in addition, the condition $\lim\limits_{n\to\infty }\left( \sum\limits_{i=1}^{\infty}\varepsilon _{n,i} \right)=0$ is satisfied.

    Consider the following sequence of simple functions $\{ S_n\}:\ {\mathbb N}\to S_2(E,{\cal R}_{\cal E},\lambda _{\cal E},{\mathbb C})$:
    $$
    S_n=\prod\limits_{i=1}^{\infty }s_{n,i},
    $$
    where $s_{n,i}=\chi _{[a_i,b_i]}$ when $i>n$ and otherwise $s_{n,i}$ is a simple function such that inequalities $\| s_{n,i}-f_i\|_{L_2({\mathbb R})}<\varepsilon _{n,i}$ and $\| s_{n,i}-f_i\|_{L_1({\mathbb R})}<\varepsilon _{n,i}$ are satisfied.

    The functions $s_{n,i},\, n,i\in\mathbb N$ can be chosen in such a way that for all $m,n\in\mathbb N$ such that $n<m$, and for almost all $x\in\mathbb R$ the condition $|s_{n,i}(x)|\leq|s_{m,i}(x)|\leq|f_i(x)|$ is satisfied. This allows from the inequality $\prod\limits_{i=m_1}^{m_2}\|f _i\|_{L_2({\mathbb R})}^2\leq C\ \forall \ m_1,m_2\in {\mathbb N}$ (which follows from the unconditional convergence of the product $\prod\limits_{k=1}^{\infty }\|f _k\|_{L_2({\mathbb R})}^2$) to get an estimate $\prod\limits_{i=m_1}^{m_2}\|s _{n,i}\|_{L_2({\mathbb R})}^2\leq C\ \forall \ m_1,m_2\in {\mathbb N}$ with constant $C$ independent of $m_1,m_2,n\in \mathbb N$.

    It is easy to see that $S_n\in S_2(E,{\cal R}_{\cal E},\lambda_{\cal E},{\mathbb C})$ for each $n\in\mathbb N$. We prove that the sequence $\{S_n\}$ is fundamental in the space ${\cal H}_{\cal E}$ and its limit has the property (\ref{prod}).

    Let $n,m\in\mathbb{N}$ and $n<m$. Then according to lemma \ref{lemma-inf-prod-ineq} we get that
\begin{eqnarray*}
\| S_m-S_n\|_{{\cal H}_{\cal E}}
    &=&   \left\| \left(
            \prod\limits_{k=1}^ms_{m,k}- \prod\limits_{k=1}^ms_{n,k}
        \right)
    \prod\limits_{k=m+1}^{\infty }\chi _{[a_k,b_k]} \right\| _{{\cal H}_{\cal E}}\\
    &=&   \left\| 
            \left(
                \prod\limits_{k=1}^ms_{m,k}- \prod\limits_{k=1}^ms_{n,k}
            \right)
        \right\|_{L_2(R^m)}
        \left(
            \prod\limits_{k=m+1}^{\infty }(b_k-a_k)
        \right)^{1/ 2} \\
        &\leq& C^2\sum\limits_{k=1}^{\infty } 
        \left\|
            s_{m,k}-f_k+f_k-s_{n,k}
        \right\|_{L_2({\mathbb R})}
        \left(
            \prod\limits_{k=m+1}^{\infty }(b_k-a_k)\right)^{1/2}  \\
            &\leq& C^2
        \left[
            \sum\limits_{k=1}^n(\varepsilon _{m,k}+\varepsilon _{n,k})+\sum\limits_{k=n+1}^m(\varepsilon _{m,k}+\delta _{k})
        \right]
        \left(
            \prod\limits_{k=m+1}^{\infty }(b_k-a_k)
        \right)^{1/2} .
\end{eqnarray*}
Thus $\{S_n\}$ is fundamental. 

Denote by $F\in {\cal H}_{\cal E}$ the limit of the sequence $\{S_n\}$. Then for any absolutely measurable block $\Pi_{c,d}$, the equality $\lim\limits_{n\to\infty}|(F,\chi_{\Pi_{c,d}})-(S_n,\chi_{\Pi _{c,d}})|=0$ holds. Therefore, to prove equality (\ref{prod}) it is required to show that the sequence $\{(S_n,\chi _{\Pi _{c,d}})\}$ converges exactly to $\prod\limits_{k=1}^{\infty}\int\limits_{c_k}^{d_k}f_k(x_k)dx_k$.

Since the product $\prod\limits_{k=1}^{\infty }\|f_k\|_{L_2({\mathbb R})}^2$ converges and the block $\Pi _{c,d}$ is absolutely measurable, then there exists a number $C\geq 1$ such that
$$
    \sup\limits_{K_1,K_2}
    \sum\limits_{k=K_1}^{k=K_2}
    \left|
        \int\limits_{c_k}^{b_k}f_k(x)dx
    \right|
    \leq  C.
$$
Since $|s_{n,k}(x)|\leq|f_k(x)|$ is almost everywhere, for any $n$ holds
$$
    \sup\limits_{K_1,K_2}
    \sum\limits_{k=K_1}^{k=K_2}
    \left|
        \int\limits_{c_k}^{b_k}s_{n,k}(x)dx
    \right|
    \leq C.
$$
For any $n,k$, by virtue of choosing the sequence $\{S_n\}$, the following inequalities are satisfied
$$
    \left|
        \int\limits_{c_k}^{d_k}f_k(x)dx
        -\int\limits_{c_k}^{d_k}s_{n,k}(x)dx
    \right|
    \leq 
    \|f_k-s_{n,k}\|_{L_1({\mathbb R})}
    \leq 
    \varepsilon_{n,k}.
$$
Therefore, according to lemma \ref{lemma-inf-prod-ineq}, 
$$
    \left|
        \prod\limits_{k=1}^{\infty }
        \int\limits_{c_k}^{d_k} f_k(x)dx - \prod\limits_{k=1}^{\infty } \int\limits_{c_k}^{d_k}s_{n,k}(x)dx
    \right|
    \leq C^2\sum\limits_{k=1}^{\infty }\epsilon _{n,k},
$$
so the sequence $(S_n,\chi_{\Pi_{c,d}})$ indeed converges to $(F,\chi_{\Pi_{c,d}})$.
\end{proof}

We investigate the question of the existence of the limit of the sequence of functions $\{{\cal F}_nu\}$ for an arbitrary $u\in {\cal H}_{\cal E}$, where ${\cal F}_n$ is the Fourier transform over the first $n$ variables.
Let $\Pi$ be a block with edges $[a_k,b_k)=[-\frac{1}{2},\frac{1}{2})$ for all $k\in\mathbb N$. Let us find $\lim\limits_{n\to \infty }{\cal F}_n\chi _{\Pi}$.

\begin{example}\label{example-fourier-sin}
    Let $f(x)=\frac{2}{x}\sin (\frac{x}{2}),\, x\in \mathbb R$. Note that 
    $$
    ({\cal F}_n\chi _{\Pi })(x)=\prod\limits_{j=1}^n f(x_j)\otimes \chi _{E^n,\Pi ^n}(x_{n+1},\ldots),\, x\in E
    $$
    for all $n\in \mathbb N$. Then for any non-empty absolutely measurable block $\Pi _{a,b}$ with positive measure it is true that
    
    \begin{equation}\label{51}
    \lim\limits_{n\to \infty }(\chi _{\Pi _{a,b}},{\cal F}_n\chi _{\Pi })=0,
    \end{equation} 
    because there exists $\sigma \in (0,1)$ such that $|\int\limits_{a_k}^{b_k}f(x_k)dx_k|\leq 1-\sigma <1$ for any $k\in \mathbb N$ satisfying $b_k-a_k>\frac{1}{2}$. Since number of such $k\in \mathbb N$ for block $\Pi _{a,b}$ with positive measure is infinite, equality~(\ref{51}) is fair. 
    Since equality (\ref{51}) is fair for arbitary block $\Pi _{a,b}$ with positive measure and linear subspace  $S_2(E,{\cal R}_{\cal E},\lambda _{\cal E},{\mathbb C})$ dense in ${\cal H}_{\cal E}$, one has $\lim\limits_{n\to \infty }{\cal F}_n\chi _{\Pi}=0$.
\end{example}

\begin{theorem}
    For any $u,v\in {\cal H}_{\cal E}$, one has the limit  $\lim\limits_{n\to \infty }({\cal F}_nv,u)=0$.
\end{theorem} 

\begin{proof}
    As in the example \ref{example-fourier-sin}, it can be shown that if $\Pi ,Q$ are certain absolutely measurable block of positive measure, then
    $\lim\limits_{n\to \infty }({\cal F}_n(\chi _{\Pi }),\chi Q)=0$.
    Therefore, $\lim\limits_{n\to \infty }({\cal F}_nv,u)=0$ for any $u,v\in S_2(E,{\cal R}_{\cal E},\lambda _{\cal E},{\mathbb C})$. Since the linear manifold $S_2(E,{\cal R}_{\cal E},\lambda _{\cal E},{\mathbb C})$ is dense in space ${\cal H}_{\cal E}$ and $\| {\cal F}_n\|_{B({\cal H}_{\cal E})}=1$ for any $n\in \mathbb N$, then $\lim\limits_{n\to \infty }({\cal F}_nv,u)=0$  for any $u,v\in {\cal H}_{\cal E}$. 
\end{proof}

\begin{remark}\label{remark-fourier-can-exists}
    The fact that $\lim\limits_{n\to \infty }({\cal F}_nu,v)=0$ for any $u,v\in {\cal H}_{\cal E}$ perhaps may mean that the image of the subspace ${\cal H}_{\cal E}$ under the action of the Fourier transform on the space $\cal H$ can be some subspace in ${\cal H}\ominus {\cal H}_{\cal E}$.
\end{remark}

\section{Taylor's formula for smooth functions}
\label{Sec:Taylor}

In this section we prove Taylor's formula for shift $S_{th}$ of smooth functions (def. \ref{def-smooth-functions}) over shifts along not finite vectors under some conditions. We also prove that the expectation of random shift along $th$ with respect to an arbitrary distribution with zero expectation and variance $D$ is differentiable at $t=0$ and find its derivative. All these results are also proved for shifts in momentum representation.

\begin{lemma}\label{lemma-taylor}
    Let $D$ be a non-degenerate positive-definite nuclear operator with its own basis $\mathcal{E}$, such that $D^{1/2}$ is also nuclear. Let $h\in L_1(\mathcal{E})$ be such a vector that the vector $D^{-1/2} h$ also belongs to $L_1(\mathcal{E})$. Then for the function $f = \mathcal{U}_{s D} u(x) \in C_D^\infty$ one has
    \begin{equation}\label{eq-taylor}
        S_{th} f(x) 
            = \sum\limits_{k=0}^n \dfrac{t^k}{k!} \d^k f(x) (h, \ldots, h)+ r_{n+1},
    \end{equation}
    where 
    $$
    \d^k f(x) (h, \ldots, h) 
    = \sum\limits_{j_1, \ldots, j_k \in \mathbb{N}} 
    (h_{j_1} \ldots h_{j_k}) 
    \dfrac{\d^k}{\d x_{j_1} \ldots \d x_{j_k}} f(x)
    $$
    and 
    $$\|r_{n+1}\| _{\mathcal{H}_\mathcal{E}}
    \le \dfrac{t^{n+1}}{s^{(n+1)/2} (n+1)!} 
    \|D^{-1/2} h\|_{L_1(\mathcal{E})}^{n+1} 
    \|u\| _{\mathcal{H}_\mathcal{E}}.
    $$
\end{lemma}

\begin{proof}
    Let $h^{(m)} = P_m h$. Consider the functions
    $$
    \phi(t) = f(x + th) = S_{th} f(x), \quad \phi_m(t) = f(x + t h^{(m)}) = S_{th^{(m)}} f(x).
    $$
    By virtue of the \ref{theorem-shift-appr} theorem, pointwise convergence is performed:
    $$
    \forall t \ \lim\limits_{m\to\infty} \|\phi_m(t) - \phi(t)\| = 0.
    $$
    According to lemma \ref{lemma-der-ineq}, all mixed derivatives along coordinate directions are defined for $f(x)$, and the following inequality holds
    $$
    \left\| \dfrac{\d^k}{\d x_{j_1} \ldots \d x_{j_k}} f(x) \right\| \le \dfrac{1}{\sqrt{s^k \d_{j_1} \ldots \d_{j_k}}} \|u\|.
    $$
    In this case, the differentials of any order of functions $\phi_m$ form a fundamental sequence in the norm
\begin{eqnarray*}
\left \|\dfrac{\d^k}{\d t^k} \phi_m(t) - \dfrac{\d^k}{\d t^k} \phi_{m+p}(t)\right\| 
    &\le& 
    t^k \sum\limits
    _{  
        \substack{j_1, \ldots, j_k = 1\\
        \exists j_k \ge m+1}}
    ^{m+p} 
    |h_{j_1}| \ldots |h_{j_k}|  \cdot \left\| \dfrac{\d^k}{\d x_{j_1} \ldots \d x_{j_k}} f(x) \right\|\\
   & \le&
    t^k \left(\sum\limits_{k=1}^{m+p} \dfrac{|h_k|}{\sqrt{s d_k}} \right)^k \|u\|
    \le t^k s^{-k/2} \|D^{-1/2} h\|^k_{L_1(\mathcal{E})}\|u\|.
 \end{eqnarray*}
    Therefore, by lemma \ref{lemma-uniform-der-conv} one can prove by induction that the function $\phi(t)$ is also infinitely differentiable, and
    $$
    \dfrac{d^k}{\d t^k} \phi(t) = t^k \sum\limits_{j_1, \ldots, j_k \in \mathbb{N}} (h_{j_1} \ldots h_{j_k}) \dfrac{\d^k}{\d x_{j_1} \ldots \d x_{j_k}} f(x).
    $$
    For $S_{th^{(m)}} f(x)$, one can write the Taylor formula with a residual term in the Lagrange form (theorem~12.4.4 in~\cite{BogachevSmolyanov}):
    $$
    f(x - t h^{(m)}) 
    = \sum\limits_{k=0}^n \dfrac{t^k}{k!} \d^k f(x) (h^{(m)}, \ldots, h^{(m)})
    + \dfrac{t^{n+1}}{n!} \int\limits_0^1 (1-\theta)^n d^n f(x + \theta t h^{(m)})(h^{(m)}, \ldots, h^{(m)}) d \omega.
    $$
    Thus, 
\begin{eqnarray*}
    \left\| f(x - t h^{(m)}) - \sum\limits_{k=0}^n \dfrac{1}{k!} \d^k f(x) (h^{(m)}, \ldots, h^{(m)}) \right\| 
    &\le& \dfrac{t^{n+1}}{(n+1)!} \| \d^n f(x + \theta t h^{(m)})(h^{(m)}, \ldots, h^{(m)})\|\\
    &\le&
    \dfrac{t^{n+1}}{(n+1)!} s^{-(n+1)/2} \|D^{-1/2} h\|^{n+1}_{L_1(\mathcal{E})} \|u\|.
\end{eqnarray*}
    where the estimate does not depend on $m$. Taking the limit $m\to\infty$, we obtain the Taylor formula for $\phi(t)$.
\end{proof}

\begin{lemma}\label{lemma-random-taylor}
    Let $h \in E$ be a random vector with zero expectation and a diagonal variance operator $D$, such that $D^{1/2}$ is nuclear. Let there also exist a nuclear operator $B$ which is diagonal in the basis $\mathcal{E}$ and such that its root is nuclear, $B^{-1/2} h \in L_1(\mathcal{E})$ almost surely, and $\mathbb{E}_h\|B^{-1/2}h\|^3_{L_1(\mathcal{E})} <+ \infty$. Then for any $f\in C_B^\infty$ the function
    $$
    A(t) = \mathbb{E}_h S_{th} f(x)
    $$
    is twice differentiable at zero, and 
    $$
    \|\mathbb{E}_h S_{th} f(x) - f(x) - \dfrac{1}{2} t^2 \Delta_D f(x)\| \le C t^3.
    $$
\end{lemma}
\begin{proof}
    For each fixed $h =h(\omega)$, such that $B^{-1/2} h\in L_1(\mathcal{E})$, 
    $$
    S_{th} f(x) = f(x) + t \sum\limits_{k=1}^\infty h_k \partial_k f(x) + \dfrac{t^2}{2} \sum\limits_{k,l=1}^\infty h_k h_l \partial_k \partial_l f(x) + r(t, h),
    $$
    where $r(t, h) \le \dfrac{t^3}{6} \|B^{-1/2}h\|^3_{L_1(\mathcal{E})} \|u\|$. 

    Turning to the expectation with respect to $h$, we get
    $$
    \mathbb{E}_h S_{th} f(x) = f(x) + \dfrac{t^2}{2} \sum\limits_{k=1}^\infty d_k \dfrac{d^2}{d x_k} f(x) + \mathbb{E}_h r(h,t),
    $$
    where, by the condition $\mathbb{E}_h\|B^{-1/2}h\|^3_{L_1(\mathcal{E})} <+ \infty$,
    $$
    \|\mathbb{E}_h r(h,t)\| \le Ct^3.
    $$
\end{proof}

We have proven Taylor's formula in its usual form, when the shift in the argument $x\to x+th$ is approximated by a series of differentials of the function at the point $x$. However, a partial Fourier transform can be applied to Taylor's formula for the finite $h$ and a formula can be obtained to approximate the shift in the momentum representation. Note that now we do not require $D^{1/2}$ to be nuclear.

\begin{lemma}
    Let $D$ be a non-degenerate positive-definite nuclear operator with eigenvectors collinear to basis $\mathcal{E}$. Let $a\in L_1(\mathcal{E})$ such a vector that the sequence $D^{-1/2} h$ also lies in $L_1(\mathcal{E})$. Then for the function $f =\widehat{\mathcal{U}}_{s D} u(x) \in\widehat{C}_D^\infty$ it is true that
    \begin{equation}\label{eq-mom-taylor}
        \widehat{S}_{ta} f(x) 
        = \sum\limits_{k=0}^n \dfrac{(-1)^k t^k}{k!} 
        \sum\limits_{j_1, \ldots, j_k \in \mathbb{N}} a_{j_1} \ldots a_{j_k} x_{j_1} \ldots x_{j_k} f(x)
        + r_{n+1},
    \end{equation}
    where 
    $$
    \|r_{n+1}\| 
    \le \dfrac{t^{n+1}}{s^{(n+1)/2} (n+1)!} 
    \|D^{-1/2} a\|_{L_1(\mathcal{E})}^{n+1} 
    \|u\|.
    $$
\end{lemma}

\begin{proof}
    We are conducting finite-dimensional approximations of the shift vector $a^{(m)} = P_m a$. For each finite $a$, the formula \ref{eq-mom-taylor} is correct, since it is obtained from the usual Taylor formula \ref{eq-taylor} by the Fourier transform. It follows directly from \ref{eq-mom-taylor} that for a finite $a$ the function $t \mapsto\widehat{S}_{t a} f(x)$ is infinitely differentiable, and its $k$-th derivative is equal to
    $$
    (-1)^k \sum\limits_{j_1, \ldots, j_k \in \mathbb{N}} a_{j_1} \ldots a_{j_k} x_{j_1} \ldots x_{j_k} f(x).
    $$
    Let us introduce auxiliary functions
    $$
    \phi(t) = f(x + th) = S_{th} f(x), \quad \phi_m(t) = f(x + t h^{(m)}) = S_{th^{(m)}} f(x).
    $$
    Taking advantage of inequality
\begin{eqnarray*}
    \|x_{j_1} \ldots x_{j_k} f(x)\| 
    &=& \left\|
        \dfrac{d^k}{\d x_{j_1} \ldots dx_{j_k}} \mathcal{F}_J (f(x))
        \right\| 
    = \left\|
        \dfrac{d^k}{\d x_{j_1} \ldots \d x_{j_k}} \mathcal{U}_{sD} \mathcal{F}_J (u(x))
        \right\| \\
    &\le& \dfrac{1}{\sqrt{s^k d_{j_1} \ldots d_{j_k}}} \|\mathcal{F}_J (u(x))\|
    = \dfrac{1}{\sqrt{s^k d_{j_1} \ldots d_{j_k}}} \|u(x)\|, \quad J = \max\limits_j j_k,
\end{eqnarray*}
    we can prove the fundamentality of the sequence of $k$-th derivatives of $\phi_m$
\begin{eqnarray*}
    \left \|\dfrac{\d^k}{\d t^k} \phi_m(t) - \dfrac{\d^k}{\d t^k} \phi_{m+p}(t)\right\| 
    &\le&
    t^k \sum\limits
    _{  
        \substack{j_1, \ldots, j_k = 1\\
        \exists j_k \ge m+1}}
    ^{m+p} 
    |a_{j_1}| \ldots |a_{j_k}|  \cdot \left\| x_{j_1} \ldots x_{j_k} f(x) \right\|\\
    &\le&
    t^k \left(\sum\limits_{k=1}^{m+p} \dfrac{|a_k|}{\sqrt{s d_k}} \right)^k \|u\|
    \le t^k s^{-k/2} \|D^{-1/2} h\|^k_{L_1(\mathcal{E})}\|u\|.
 \end{eqnarray*}
    Here lemma \ref{lemma-uniform-der-conv} guarantees us infinite differentiability of $\phi(t)$ and then we can write the Taylor formula for $\phi_m(t)$ as in the proof of lemma \ref{lemma-taylor} and go to the limit of $m\to\infty$, which gives us \ref{eq-mom-taylor}
\end{proof}

\begin{lemma}\label{lemma-random-mom-taylor}
    Let $a \in E$ be a random vector with zero expectation and a nuclear variance operator $D$. Let there also be a nuclear operator $B$, diagonal in the basis $\mathcal{E}$, such that $B^{-1/2} h\in L_1(\mathcal{E})$ almost surely and $\mathbb{E}_h\|B^{-1/2}h\|^3_{L_1(\mathcal{E})} <+ \infty$. Then for any $f\in C_B^\infty$ the function
    $$
    A(t) = \mathbb{E}_h \widehat{S}_{th} f(x)
    $$
    is twice differentiable at zero, and
    $$
    \|\mathbb{E}_h \widehat{S}_{th} f(x) - f(x) - \dfrac{1}{2} t^2 (Dx,x) f(x)\| \le C t^3.
    $$
\end{lemma}
\begin{proof}
    The proof is similar to lemma \ref{lemma-random-taylor}.

    For each fixed $a = a(\omega)$, such that $B^{-1/2} a\in L_1(\mathcal{E})$, 
    $$
    \widehat{S}_{th} f(x) = f(x) + t \sum\limits_{k=1}^\infty a_k x_k f(x) + \dfrac{t^2}{2} \sum\limits_{k,l=1}^\infty a_k a_l x_k x_l f(x) + r(t, h),
    $$
    where $r(t, h) \le \dfrac{t^3}{6} \|B^{-1/2}h\|^3_{L_1(\mathcal{E})} \|u\|$.

    Moving on to the expectation over $h$, we prove the lemma.
\end{proof}

\section{Approximation of the evolution of an infinite-dimensional quantum oscillator using quantum random walks}
\label{Sec:InfiniteApproximation}

In this section we prove that evolution of diffusion process on the position space $E$ and evolution of infinite-dimensional quantum oscillator on $E$ may be approximated by composition of random shifts in coordinate and momentum representations. Thus, we generalize results of section \ref{Sec:FiniteApproximation} on infinite-dimensional case.

Let us first consider the equation
    \begin{equation}\label{eq-inf-poiss}
    \dfrac{d}{dt} u(t,x) = \dfrac{1}{2} \Delta_D u(t,x), \quad t > 0
    \end{equation}
on $E$ with an initial condition
    \begin{equation}\label{eq-inf-poiss-1}
    u(0,x) = u_0.
    \end{equation}
The operator $\Delta_D$ is self-adjoint, defined on a dense subspace of the space $\mathcal{H}_\mathcal{E}$ and is a generator of the contraction semigroup of shift averages by argument $\mathcal{U}_{tD}$. Therefore, the semigroup $\mathcal{U}_{(t/2)D}$ represents the solution of the equation (\ref{eq-inf-poiss})-(\ref{eq-inf-poiss-1}) in the same sense as in the definition of \ref{def-sol}: for any function $u_0 \in\mathcal{H}_\mathcal{E}$, the solution is $u(t,x) = \mathcal{U}_{(t/2)D} u_0(x)$.

Let us consider, by analogy with the finite-dimensional case, the mathematical expectation of a random shift.

\begin{theorem}
    Let $h\in L_1(\mathcal{E})$ be a random vector with zero expectation and diagonal variance operator $D$. Let also assume that for some nuclear diagonal operator $B$ all operators $D^{1/2}, B^{1/2}$ and $DB^{-1}$ are nuclear, and $B^{-1/2} h\in L_1(\mathcal{E})$n is fulfilled almost surely and $\mathbb{E}_h\|B^{-1/2} h\|^3_1 \le + \infty$.

    Then
    $$
    \mathbb{E}_h \Av_m S_{\sqrt{t}h}
    $$
    converges uniformly on the segments in a strong topology to the semigroup $\mathcal{U}_{(t/2)D}$, which represents the solution of the equation (\ref{eq-inf-poiss})-(\ref{eq-inf-poiss-1}).
\end{theorem}

\begin{proof}
    Consider the operator-valued function $T(t)$, acting according to the rule
    $$
    T(t) f(x) = \mathbb{E}_h S_{h \sqrt{t}} f(x).
    $$
    Due to the independence of the shift operators and lemma \ref{expect-independent-unitary},
    $$
    \mathbb{E}_h \Av_m S_{\sqrt{t}h} = (T(h/m))^m.
    $$
    According to lemma \ref{lemma-random-taylor}, the function $T(t) f(x)$ is differentiable at zero for any function from the dense linear subspace $C_B^\infty$, and its derivative is $\dfrac{1}{2} \Delta_D$, the closure of which, by the theorem \ref{th-delta-core}, is a generator of the semigroup of the averaging function over the Gaussian measure $\mathcal{U}_{(t/2)D}$.

    It is also obvious that $T(0) = 1$, and $\|T(t)\|\le 1$.  Therefore, we can use Chernov's theorem.
\end{proof}

Let us now consider the problem on $E$
\begin{equation}\label{eq-inf-osc}
\dfrac{d}{dt} u(t,x) = \dfrac{1}{2} \Delta_{D_x} u(t,x) - \dfrac{1}{2} (D_p x,x) u(t,x), \quad t > 0,
\end{equation}
\begin{equation}\label{eq-inf-osc-1}
u(0,x) = u_0.
\end{equation}
where the operators $D_x$ and $D_p$ are nuclear, diagonal in the basis $\mathcal{E}$ and positive-definite. Let there also be such a nuclear operator $B$, diagonal in the basis $\mathcal{E}$, that $D_x^{1/2}$, $B^{1/2}$, $D_x B^{-1}$ and $D_p B^{-1}$ are also nuclear. The diagonal elements of these operators are denoted by $d_{x,k}$, $d_{p,k}$ and $b_k$, respectively. The theorems \ref{th-sum-core} under such conditions guarantees us that the closure of the positive linear combination $\Delta_{D_x}$ and $-(D_p x,x)$ is a self-adjoint operator generating some semigroup. Therefore, we can use Chernov's theorem to approximate the solution of this equation.

\begin{theorem}\label{th-inf-osc-appr-1}
    Let $h, a \in L_1(\mathcal{E})$ be random vectors with zero mathematical expectation and variances $D_x$ and $D_p$, which are diagonal in basis $\mathcal{E}$. Suppose there is also such a diagonal operator $B$, such that all operators $D_x^{1/2}, D_p, B^{1/2}, D_x B^{-1}$ and $D_p B^{-1}$ are nuclear, vectors $B^{-1/2} h$ and $B^{-1/2} a$ almost surely belong to $L_1(\mathcal{E})$, and conditions
    $\mathbb{E}_h \|B^{-1/2} h\|^3_{L_1(\mathcal{E})} \le + \infty$ and $\mathbb{E}_a \|B^{-1/2} a\|^3_{L_1(\mathcal{E})} \le + \infty$ are fair

    Then 
    $$
    \mathbb{E}_h \Av_m (S_{\sqrt{t}h} \widehat S_{\sqrt{t}a})
    $$
    converges uniformly on the segments in a strong topology to a semigroup which represents the solution of the equation (\ref{eq-inf-osc})-(\ref{eq-inf-osc-1}).
\end{theorem}

\begin{proof}
    If 
    $$
    V(t) = \mathbb{E}_h (S_{\sqrt{t}h} \widehat S_{\sqrt{t}a}) f(x)
    $$
    is fair for any $f$ from 
    $$
    D_B = \{\widehat{\mathcal{U}}_{sB} \mathcal{U}_{tB} u \mid u \in \mathcal{H}_\mathcal{E}, t>0, s>0\},
    $$
    then 
    $$
    \mathbb{E}_h \Av_m (S_{\sqrt{t}h} \widehat S_{\sqrt{t}a}) f(x) = (V(t/m))^m.
    $$
    Obviously,
    $V(0) = I$, $\|V(t)\| \le 1$ and $V'(0) = \dfrac{1}{2} \Delta_{D_x} f(x) - \dfrac{1}{2} (D_px,x) f(x)$.

    By the theorem \ref{th-sum-core}, the closure of the operator $\dfrac{1}{2} \Delta_{D_x} -\dfrac{1}{2} (D_px,x): D_B\to \mathcal{H}_\mathcal{E}$ is self-adjoint, and, therefore, is a generator of some self-adjoint semigroup.

    Chernov's theorem gives the required result.
\end{proof}

The following theorem is proved in a similar way.

\begin{theorem}\label{th-inf-osc-appr-2}
   Let $h, a \in L_1(\mathcal{E})$ be random vectors with zero mathematical expectation and variances $D_x$ and $D_p$, which are diagonal in basis $\mathcal{E}$. Suppose there exists a diagonal operator $B$ such that all operators $D_x^{1/2}, D_p, B^{1/2}, D_x B^{-1}$ and $D_p B^{-1}$ are nuclear, vectors $B^{-1/2} a$ and $B^{-1/2} h$ almost surely belong to $L_1(\mathcal{E})$, and conditions 
    $\mathbb{E}_h \|B^{-1/2} h\|^3_{L_1(\mathcal{E})} \le + \infty$ and $\mathbb{E}_a \|B^{-1/2} a\|^3_{L_1(\mathcal{E})} \le + \infty$ are fair.

    Denote by $T_{p,h,a}$ a random operator which with probability $p$ makes a shift in the coordinate representation of $S_{h}$, and otherwise ~--- a shift in the momentum representation of $\widehat{S}_a$.

    Then the sequence
    $$
    \mathbb{E}_h \Av_m T_{p,\sqrt{t}h,\sqrt{t}a}
    $$
    converges uniformly along the segments in a strong topology to a semigroup which represents the solution of the equation (\ref{eq-inf-osc})-(\ref{eq-inf-osc-1}).
\end{theorem}

\section{Conclusions}\label{Sec:Conclusions}

We consider quantum random walks in an infinite-dimensional phase space which is constructed using Weyl representation of the coordinate and momentum operators in the space of functions on a Hilbert space which are square integrable with respect to a shift-invariant measure. 
We prove that evolution of an infinite-dimensional quantum oscillator can be approximated by a quantum random walks in this phase space. 

The main tool to obtain the Weyl representation of quantum random walks in the phase space is the construction of a non-negative shift-invariant measure on the position subspace of the phase space. This measure is realised as an finitely-additive measure since according to Weil theorem it can not have all properties of the Lebesgue measure. We analyse operators of the Weyl representation and other linear operators in the space of square integrable with respect to this invariant measure functions. We show that the averaging of operator by shift of the argument by a Gaussian random vector is described by a convolution with a Gaussian measure. We obtain a criterion for  strong continuity of the semigroup of self-adjoint contractions which is generated by the convolutional semigroup of  the Gaussian measures. We prove that the generator of this semigroup is a self-adjoint Laplace-Volterra operator in the space of square integrable functions. We also study smoothness of the images of functions under averaging of the random shifts.

By using the Weyl representation of the shift operators in the momentum space we obtain that averaging of the operators of shift  on a Gaussian random vector in momentum space is presented by operators of multiplication by a Gaussian function in the space of square integrable with respect to the invariant measure functions. The generator of the semigroup of such multiplication operators is described. Further we study properties of the generator of the semigroup which is given by averaging of the Weyl representation of random shifts in the phase space. We prove Taylor's formula for shift of argument of smooth functions along a vector from a certain subspace in coordinate and momentum representations.  We prove  that positive combinations of generators of the semigroup of convolutions with Gaussian measures and the semigroup of multiplication by Gaussian functions under some conditions are self-adjoint and generate strongly continuous semigroups. This allows us to define a Hamiltonian of an infinite-dimensional quantum oscillator.

The mathematical tools that are developed in this paper can be extended to study other models of quantum random walks that are given by Koopman representation of groups of transformations of an infinite dimensional phase spaces. In particular, it can be applied for studying random walks along random Hamiltonian vector fields in the Weyl representation.


\begin{thebibliography}{}

\bibitem[Accardi et~al., 2002]{Accardi_Pechen_Volovich_2002}
Accardi, L., Pechen, A.~N., and Volovich, I.~V. (2002).
\newblock Quantum stochastic equation for the low density limit.
\newblock {\em Journal of Physics A: Mathematical and General},
  35(23):4889–4902.

\bibitem[Aharonov et~al., 1993]{Aharonov_Davidovich_Zagury_1993}
Aharonov, Y., Davidovich, L., and Zagury, N. (1993).
\newblock Quantum random walks.
\newblock {\em Physical Review A}, 48(2):1687.

\bibitem[Aizenman and Simone, 2015]{Aizenman_Simone_2015}
Aizenman, M. and Simone, W. (2015).
\newblock {\em Random operators. Disorder effects on quantum spectra and
  dynamics}.
\newblock Graduate Studies in Mathematics. American Mathematical Society,
  Providence, RI.

\bibitem[Attal and Dhahri, 2010]{Attal_Dhahri_2010}
Attal, S. and Dhahri, A. (2010).
\newblock Repeated quantum interactions and unitary random walks.
\newblock {\em Journal of Theoretical Probability}, 23:345--361.

\bibitem[Attal et~al., 2012]{Attal_Petruccione_Sabot_Sinayskiy_2012}
Attal, S., Petruccione, F., Sabot, C., and Sinayskiy, I. (2012).
\newblock Open quantum random walks.
\newblock {\em Journal of Statistical Physics}, 147(4):832--852.

\bibitem[Baker, 1991]{Baker}
Baker, R. (1991).
\newblock “{L}ebesgue measure” on $\mathbb{R}^\infty$.
\newblock {\em Proceedings of the American Mathematical Society},
  113(4):1023--1029.

\bibitem[Belton, 2010]{Belton_2010}
Belton, A.~C. (2010).
\newblock Quantum random walks and thermalisation.
\newblock {\em Communications in Mathematical Physics}, 300:317--329.

\bibitem[Bogachev, 1998]{BogGM}
Bogachev, V.~I. (1998).
\newblock {\em Gaussian measures}.
\newblock American Mathematical Soc., New {Y}ork.

\bibitem[Bogachev and Smolyanov, 2020]{BogachevSmolyanov}
Bogachev, V.~I. and Smolyanov, O.~G. (2020).
\newblock {\em Real and functional analysis}.
\newblock Springer, Moscow.

\bibitem[Bogol{\^u}bov et~al., 2012]{Bogo}
Bogol{\^u}bov, N.~N., Logunov, A.~A., Oksak, A., and Todorov, I. (2012).
\newblock {\em General principles of quantum field theory}, volume~10.
\newblock Springer Science \& Business Media, Moscow.

\bibitem[Busovikov and Sakbaev, 2020]{BusS}
Busovikov, V.~M. and Sakbaev, V.~Z. (2020).
\newblock Sobolev spaces of functions on a {H}ilbert space endowed with a
  translation-invariant measure and approximations of semigroups.
\newblock {\em Izvestiya: Mathematics}, 84(4):694.

\bibitem[Chaturvedi and Srivastava, 1983]{Chaturvedi_Srivastava_1983}
Chaturvedi, M. and Srivastava, V. (1983).
\newblock Random‐walk theory for localization.
\newblock {\em International Journal of Quantum Chemistry}, 23(4):1463--1468.

\bibitem[Dhamapurkar and Dahlsten, 2023]{Dhamapurkar_Dahlsten_2023}
Dhamapurkar, S. and Dahlsten, O. (2023).
\newblock Quantum walks as thermalizations, with application to fullerene
  graphs.
\newblock {\em arXiv preprint arXiv:2304.01572}.

\bibitem[Engel et~al., 2000]{engel-nagel}
Engel, K.-J., Nagel, R., and Brendle, S. (2000).
\newblock {\em One-parameter semigroups for linear evolution equations}, volume
  194.
\newblock Springer, New {Y}ork.

\bibitem[Filippov et~al., 2020]{Filippov2020}
Filippov, S.~N., Semin, G.~N., and Pechen, A.~N. (2020).
\newblock Quantum master equations for a system interacting with a quantum gas
  in the low-density limit and for the semiclassical collision model.
\newblock {\em Physical Review A}, 101(1):012114.

\bibitem[Gough et~al., 2021]{GOSS21}
Gough, J., Orlov, Y.~N., Sakbaev, V.~Z., and Smolyanov, O.~G. (2021).
\newblock Random quantization of {H}amiltonian systems.
\newblock In {\em Doklady Mathematics}, volume 103, pages 122--126. Springer.

\bibitem[Gough et~al., 2022]{GOSS22}
Gough, J., Orlov, Y.~N., Sakbaev, V.~Z., and Smolyanov, O.~G. (2022).
\newblock Markov approximations of the evolution of quantum systems.
\newblock {\em Doklady Mathematics}, 105(2):92--96.

\bibitem[Holevo, 2011]{H}
Holevo, A.~S. (2011).
\newblock {\em Probabilistic and statistical aspects of quantum theory},
  volume~1.
\newblock Springer Science \& Business Media, Moscow.

\bibitem[Joye, 2011]{Joye_2011}
Joye, A. (2011).
\newblock Random time-dependent quantum walks.
\newblock {\em Communications in {M}athematical {P}hysics}, 307(1):65--100.

\bibitem[Joye, 2012]{Joye_2012a}
Joye, A. (2012).
\newblock Dynamical localization for d-dimensional random quantum walks.
\newblock {\em Quantum Information Processing}, 11(5):1251--1269.

\bibitem[Joye and Merkli, 2010]{Joye_Merkli_2010}
Joye, A. and Merkli, M. (2010).
\newblock Dynamical localization of quantum walks in random environments.
\newblock {\em Journal of Statistical Physics}, 140(6):1025--1053.

\bibitem[Kato, 2013]{kato72}
Kato, T. (2013).
\newblock {\em Perturbation theory for linear operators}, volume 132.
\newblock Springer Science \& Business Media, Berlin.

\bibitem[Kempe, 2003]{Kempe_2003}
Kempe, J. (2003).
\newblock Quantum random walks: an introductory overview.
\newblock {\em Contemporary Physics}, 44(4):307--327.

\bibitem[Kempe, 2009]{Kempe_2009}
Kempe, J. (2009).
\newblock Quantum random walks: an introductory overview.
\newblock {\em Contemporary Physics}, 50(1):339--359.

\bibitem[Kolokoltsov, 2020]{Kolokoltsov_2020}
Kolokoltsov, V.~N. (2020).
\newblock Continuous time random walks modeling of quantum measurement and
  fractional equations of quantum stochastic filtering and control.
\newblock {\em arXiv preprint arXiv:2008.07355}.

\bibitem[Kuo, 1975]{Kuo75}
Kuo, H. (1975).
\newblock {\em Gaussian Measures in {B}anach Spaces, Lectures Notes in
  Mathematics}.
\newblock Springer Berlin, Heidelberg, Berlin.

\bibitem[Oksak and Sukhanov, 2003]{Oksak}
Oksak, A.~I. and Sukhanov, A.~D. (2003).
\newblock Representation of quantum {B}rownian motion in the collective
  coordinate method.
\newblock {\em Theoretical and {M}athematical {P}hysics}, 136:994--1021.

\bibitem[Orlov et~al., 2020]{OSZ-20}
Orlov, Y.~N., Sakbaev, V.~Z., and Zavadsky, D. (2020).
\newblock Operator random walks and quantum oscillator.
\newblock {\em Lobachevskii Journal of Mathematics}, 41:676--685.

\bibitem[Pechen, 2004]{PechenJMP2004}
Pechen, A.~N. (2004).
\newblock Quantum stochastic equation for a test particle interacting with a
  dilute {B}ose gas.
\newblock {\em Journal of Mathematical Physics}, 45(1):400–417.

\bibitem[Pechen and Volovich, 2002]{Pechen_Volovich_2002}
Pechen, A.~N. and Volovich, I.~V. (2002).
\newblock Quantum multipole noise and generalized quantum stochastic equations.
\newblock {\em Infinite Dimensional Analysis, Quantum Probability and Related
  Topics}, 05(04):441–464.

\bibitem[Reed and Simon, 1975]{RS2}
Reed, M. and Simon, B. (1975).
\newblock Methods of modern mathematical physics {II}, {F}ourier analysis,
  selfadjointness.

\bibitem[Sahu, 2008]{Sahu_2008}
Sahu, L. (2008).
\newblock Quantum random walks and their convergence to {E}vans-{H}udson flows.
\newblock {\em Proceedings Mathematical Sciences}, 118(3):443--465.

\bibitem[Sakbaev, 2017]{S2016}
Sakbaev, V. (2017).
\newblock Averaging of random walks and shift-invariant measures on a {H}ilbert
  space.
\newblock {\em Theoretical and Mathematical Physics}, 191(3):473--502.

\bibitem[Sakbaev, 2019]{SITO}
Sakbaev, V. (2019).
\newblock Random walks and measures on {H}ilbert space that are invariant with
  respect to shifts and rotations.
\newblock {\em Journal of Mathematical Sciences}, 241(4):469--500.

\bibitem[Sakbaev and Smolyanov, 2018]{SS18feynman}
Sakbaev, V.~Z. and Smolyanov, O.~G. (2018).
\newblock Feynman calculus for random operator-valued functions and their
  applications.
\newblock {\em Uchenye Zapiski Kazanskogo Universiteta. Seriya
  Fiziko-Matematicheskie Nauki}, 160(2):373--383.

\bibitem[Shiryaev, 2007]{shiryaev}
Shiryaev, A.~N. (2007).
\newblock {\em Probability}.
\newblock MCCME, Moscow.

\bibitem[Smolyanov and Shavgulidze, 2015]{path_integrals}
Smolyanov, O. and Shavgulidze, E. (2015).
\newblock Kontinual'nie integraly.
\newblock {\em URSS, Moscow}.

\bibitem[Smolyanov and Shamarov, 2020]{SSh}
Smolyanov, O.~G. and Shamarov, N.~N. (2020).
\newblock Schr{\"o}dinger quantization of infinite-dimensional {H}amiltonian
  systems with a nonquadratic {H}amiltonian function.
\newblock In {\em Doklady Mathematics}, volume 101, pages 227--230. Springer.

\bibitem[Sonis, 1966]{sonis}
Sonis, M.~G. (1966).
\newblock Certain measurable subspaces of the space of all sequences with a
  {G}aussian measure.
\newblock {\em Uspekhi Mat. Nauk}, 21(5 (131):277--279.

\bibitem[Vakhania et~al., 2012]{Weyl}
Vakhania, N., Tarieladze, V., and Chobanyan, S. (2012).
\newblock {\em Probability distributions on {B}anach spaces}, volume~14.
\newblock Springer Science \& Business Media, Georgia.

\bibitem[Venegas-Andraca, 2012]{Venegas-Andraca_2012}
Venegas-Andraca, S.~E. (2012).
\newblock Quantum walks: a comprehensive review.
\newblock {\em Quantum Information Processing}, 11(5):1015--1106.

\bibitem[Vershik, 2007]{Vershik}
Vershik, A.~M. (2007).
\newblock Does there exist a {L}ebesgue measure in the infinite-dimensional
  space?
\newblock {\em Proceedings of the Steklov Institute of Mathematics},
  259(1):248--272.

\bibitem[Wang and Wang, 2020]{Wang_Wang_2020}
Wang, C. and Wang, C. (2020).
\newblock Higher-dimensional quantum walk in terms of quantum {B}ernoulli
  noises.
\newblock {\em Entropy}, 22(5):504.

\bibitem[Weickert, 2001]{Weickert}
Weickert, B. (2001).
\newblock Infinite-dimensional complex dynamics: a quantum random walk.
\newblock {\em Discrete and Continuous Dynamical Systems}, 7(3):517--524.

\bibitem[Weil, 1940]{Weil}
Weil, A. (1940).
\newblock {\em L'Integration dans les groupes topologiques et ses
  applications}, volume 869.
\newblock Actualites scientifiques et industrielles, Paris.

\bibitem[Zavadskii and Sakbaev, 2019]{SZ19}
Zavadskii, D. and Sakbaev, V. (2019).
\newblock Diffusion on a {H}ilbert space equipped with a shift- and
  rotation-invariant measure.
\newblock {\em Proceedings of the Steklov Institute of Mathematics},
  306(0):112--130.

\end{thebibliography}
\end{document}